\documentclass[sigconf, nonacm]{acmart}
\newcommand\vldbdoi{10.14778/3538598.3538601}
\newcommand\vldbpages{1766 - 1778}
\newcommand\vldbvolume{15}
\newcommand\vldbissue{9}
\newcommand\vldbyear{2022}
\newcommand\vldbauthors{\authors}
\newcommand\vldbtitle{\shorttitle} 
\newcommand\vldbavailabilityurl{URL_TO_YOUR_ARTIFACTS}
\newcommand\vldbpagestyle{plain} 
\usepackage{amsmath,amsfonts}
\usepackage{graphicx}
\usepackage{textcomp}
\usepackage{xcolor}
\usepackage{booktabs}       % professional-quality
\usepackage{nicefrac}       % compact symbols for 1/2, etc.
\usepackage{microtype}      % microtypography
\usepackage{subfig}
\usepackage{xspace}
\usepackage{amsthm}
\usepackage{algorithm}
\usepackage[noend]{algpseudocode}
\usepackage{varwidth}
\usepackage{balance}
\usepackage{multirow}
\usepackage{bm}
\usepackage{dcolumn}
\usepackage{url}

\newtheorem{theorem}{Theorem}

\newtheorem{definition}{Definition}

\newcommand{\method}{\textsc{HDPView}\xspace}
\def\algbackskip{\hskip-\ALG@thistlm}

\newcommand\numberthis{\addtocounter{equation}{1}\tag{\theequation}}

\begin{document}

\title{HDPView: Differentially Private Materialized View for Exploring High Dimensional Relational Data}

\author{Fumiyuki Kato}
\affiliation{%
  \institution{Kyoto University}
}
\email{fumiyuki@db.soc.i.kyoto-u.ac.jp}

\author{Tsubasa Takahashi}
\affiliation{%
  \institution{LINE Corporation}
}
\email{tsubasa.takahashi@linecorp.com}

\author{Shun Takagi}
\affiliation{%
  \institution{Kyoto University}
}
\email{takagi.shun.45a@st.kyoto-u.ac.jp}

\author{Yang Cao}
\affiliation{%
  \institution{Kyoto University}
}
\email{yang@i.kyoto-u.ac.jp}

\author{Seng Pei Liew}
\affiliation{%
  \institution{LINE Corporation}
}
\email{sengpei.liew@linecorp.com}

\author{Masatoshi Yoshikawa}
\affiliation{%
  \institution{Kyoto University}
}
\email{yoshikawa@i.kyoto-u.ac.jp}

\begin{abstract}
How can we explore the unknown properties of high-dimensional sensitive relational data while preserving privacy?
We study how to construct an explorable privacy-preserving materialized view under differential privacy.
No existing state-of-the-art methods simultaneously satisfy the following essential properties in data exploration: workload independence, analytical reliability (i.e., providing error bound for each search query), applicability to high-dimensional data, and space efficiency.
To solve the above issues, we propose \method, which creates a differentially private materialized view by well-designed recursive bisected partitioning on an original data cube, i.e., count tensor.
Our method searches for block partitioning to minimize the error for the counting query, in addition to randomizing the convergence, by choosing the effective cutting points in a differentially private way, resulting in a less noisy and compact view.
Furthermore, we ensure formal privacy guarantee and analytical reliability by providing the error bound for arbitrary counting queries on the materialized views.
\method has the following desirable properties:
(a) \textit{Workload independence},
(b) \textit{Analytical reliability},
(c) \textit{Noise resistance on high-dimensional data},
(d) \textit{Space efficiency}.
To demonstrate the above properties and the suitability for data exploration, we conduct extensive experiments with eight types of range counting queries on eight real datasets.
\method outperforms the state-of-the-art methods in these evaluations.
% Specifically, in our experiments, the data size of the materialized view is $5\times10^{3}$ times smaller on average than the state-of-the-art data-aware segmentation method while improving the utility.
% Our proposed method helps us to explore relational data while preserving data privacy and utility.
\end{abstract}

\maketitle

\pagestyle{\vldbpagestyle}
\begingroup\small\noindent\raggedright\textbf{PVLDB Reference Format:}\\
\vldbauthors. \vldbtitle. PVLDB, \vldbvolume(\vldbissue): \vldbpages, \vldbyear.\\
\href{https://doi.org/\vldbdoi}{doi:\vldbdoi}
\endgroup
\begingroup
\renewcommand\thefootnote{}\footnote{\noindent
This work is licensed under the Creative Commons BY-NC-ND 4.0 International License. Visit \url{https://creativecommons.org/licenses/by-nc-nd/4.0/} to view a copy of this license. For any use beyond those covered by this license, obtain permission by emailing \href{mailto:info@vldb.org}{info@vldb.org}. Copyright is held by the owner/author(s). Publication rights licensed to the VLDB Endowment. \\
\raggedright Proceedings of the VLDB Endowment, Vol. \vldbvolume, No. \vldbissue\ %
ISSN 2150-8097. \\
\href{https://doi.org/\vldbdoi}{doi:\vldbdoi} \\
}\addtocounter{footnote}{-1}\endgroup

% \ifdefempty{\vldbavailabilityurl}{}{
% \vspace{.3cm}
% \begingroup\small\noindent\raggedright\textbf{PVLDB Artifact Availability:}\\
% The source code, data, and/or other artifacts have been made available at \url{\vldbavailabilityurl}.
% \endgroup

% \begin{IEEEkeywords}
% differential privacy, tensor segmentation, data exploration
% \end{IEEEkeywords}

\section{Introduction} \label{sec:intro}

In the early stage of data science workflows, exploring a database to understand its properties in terms of multiple attributes is essential to designing the subsequent tasks. 
To understand the properties, data analysts need to issue a wide variety of range counting queries.
If the database is highly sensitive (e.g., personal healthcare records), data analysts may have little freedom to explore the data due to privacy issues \cite{su2020re, wang2019surfcon}.

How can we explore the properties of high-dimensional sensitive data while preserving privacy?
This paper focuses on guaranteeing differential privacy (DP) \cite{dwork2006differential, dwork2014algorithmic} via random noise injections.
As Figure \ref{fig:overview} shows, we especially study how to construct a \textit{privacy-preserving materialized view} (\textit{p-view} for short) of relational data, which enables data analysts to explore arbitrary range counting queries in a differential private way.
Note that once a p-view is created, the privacy budget is not consumed any more for publishing counting queries, different from interactive differentially private query systems \cite{johnson2018towards, ge2019apex, wilson2019differentially, rogers2020linkedin, johnson2018chorus}, which consume the budget every time queries are issued.
In this work, we describe the desirable properties of the p-view, especially in data exploration for high-dimensional data, and fill the gaps of the existing methods.

\begin{figure}[t]
    \centering
    \includegraphics[width=0.95\hsize]{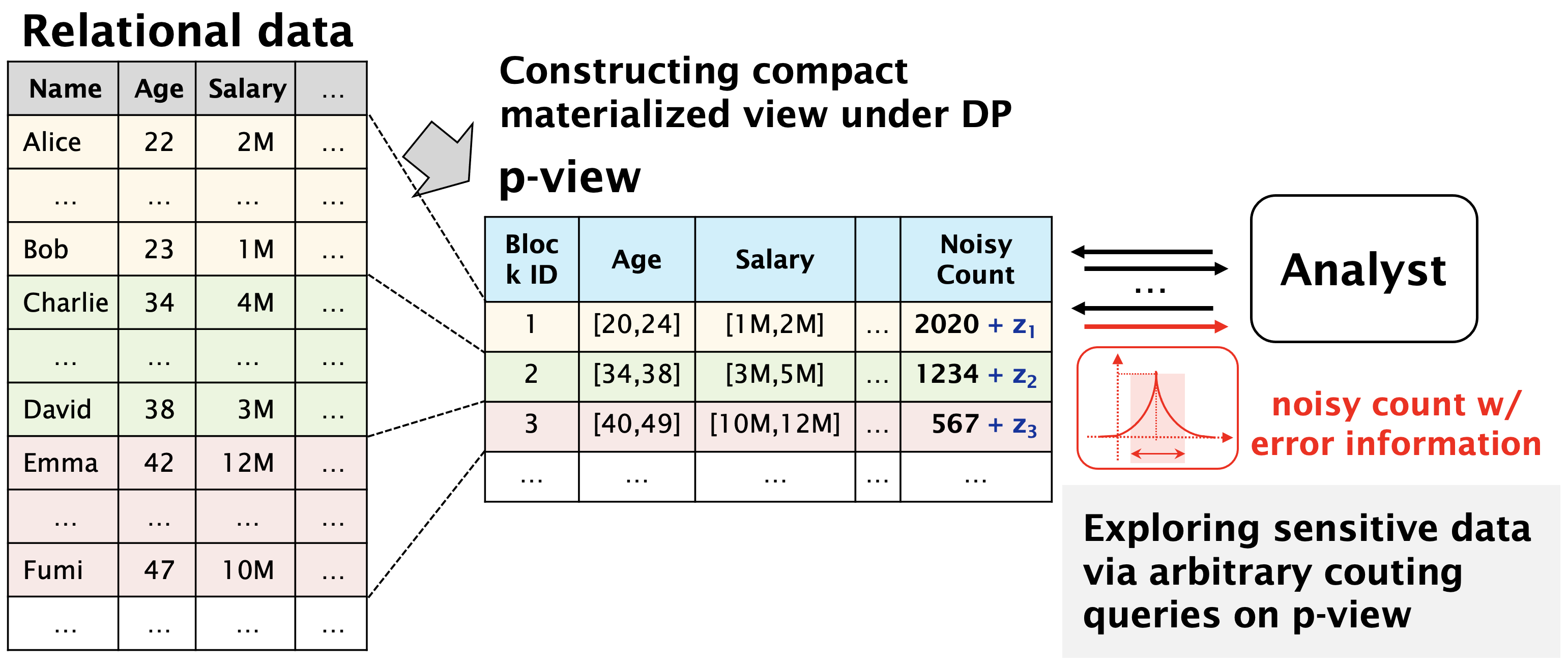}
    \caption{Data exploration through a \textit{privacy-preserving materialized view} (\textit{p-view} for short) of a multidimensional relational data. The p-view works as an independent query system. Analysts can explore sensitive and multidimensional data by issuing any range counting queries over the p-view before downstream data science workflows.}
    \label{fig:overview}
\end{figure}

\begin{table*}[t]
    \centering
    \caption{Only the proposed method achieves all requirements in private data exploration for high-dimensional data. Each competitor represents a baseline \cite{dwork2006differential}, data partitioning \cite{zhang2016privtree, li2014data, xiao2012dpcube}, workload optimization \cite{mckenna2018optimizing, li2015matrix, kotsogiannis2019privatesql}, and generative model \cite{zhang2017privbayes, qardaji2014priview, fan2020relational}, respectively.}
    \begin{tabular}{lccccc}
    \toprule
    & Identity\footnotemark[1] \cite{dwork2006differential}
    & Privtree \cite{zhang2016privtree}
    & HDMM \cite{mckenna2018optimizing}
    & PrivBayes \cite{zhang2017privbayes}
    & \textbf{\method (ours)} \\
    \midrule
    Workload independence
    & \checkmark
    & \checkmark
    & 
    & \checkmark
    & \checkmark
    % & \checkmark 
    \\
    Analytical reliability
    & \checkmark
    & \checkmark
    & \checkmark
    & 
    & \checkmark
    % & \checkmark 
    \\
    Noise resistance on high-dimensional
    & 
    & only low-dimensional
    & \checkmark
    & \checkmark
    & \checkmark
    % & \checkmark 
    \\
    Space efficiency
    & 
    & 
    & \checkmark
    & \checkmark
    & \checkmark
    % & \checkmark 
    \\
    \bottomrule
    \end{tabular}
\label{tbl:salesman}
\end{table*}

Several methods for constructing a p-view have been studied in the existing literature.
The most primitive method is to add Laplace noise \cite{dwork2006differential} to each cell of the count tensor (or vector) representing the original histogram and publish the perturbed data as a p-view.
While this noisy view can answer arbitrary range counting queries with a DP guarantee, it accumulates a large amount of noise.
Data-aware partitioning methods \cite{li2014data, Zhang2014TowardsAH, qardaji2013differentially, xiao2012dpcube, yaroslavtsev2013accurate, zhang2016privtree, kotsogiannis2019privatesql} are potential solutions, but they focus only on low-dimensional data due to the high complexity of discovering the optimal partitioning when the data have multiple attributes.
Additionally, these methods require exponentially large spaces as the dimensionality of the data increases due to the count tensor representation, which can easily make them impractical.
% Privtree \cite{zhang2016privtree} provides heuristic block partitioning for multidimensional count tensors, but it is also only suitable for low-dimensional data because of the fixed cutting method.
Workload-aware optimization methods \cite{li2014data, li2015matrix, mckenna2018optimizing, kotsogiannis2019privatesql} are promising techniques for releasing query answers for high-dimensional data; however, they cannot provide query-independent p-views needed in data exploration.

In addition, one of the most popular approaches these days is differentially private learning of generative models \cite{qardaji2014priview, chen2015differentially, zhang2017privbayes, jordon2018pate, takagi2021p3gm, harder2021dp, fan2020relational, Zhang2021PrivSynDP, ge2020kamino}.
Through the training of deep generative models \cite{jordon2018pate, takagi2021p3gm, harder2021dp, ge2020kamino, fan2020relational} or graphical models \cite{zhang2017privbayes, Zhang2021PrivSynDP}, counting queries and/or marginal queries can be answered directly from the model or indirectly with synthesized data via sampling.
These methods are very space efficient because the synthetic dataset or graphical model can be used to answer arbitrary counting queries.
However, these families rely on complex optimization methods such as DP-SGD \cite{abadi2016deep}, and it is very difficult to quantitatively estimate the error of counting queries using synthetic data, which eventually leads to a lack of reliability in practical use.
Unlike datasets often used in the literature, the data collected in the practical field may be completely unmodelable.
% Verifying generative models or synthesized data while properly protecting privacy remains open question.
% As described above, although synthetic datasets are very easy to use at first glance, they are difficult to handle in practice, and we believe that there is still a large gap between privacy-preserving data synthesis (PPDS) and data exploration on high-dimensional data.
% Therefore, we focus on extending the classical data-aware partitioning methods to allow for error analysis, giving priority to the reliability of the data exploration.
Table \ref{tbl:salesman} summarizes a comparison between the most related works and our method, \method, i.e., {\bf H}igh-{\bf D}imensional {\bf P}rivate materialized {\bf View}.
Each method is described in more detail in Section \ref{sec:related_works}.

\footnotetext[1]{Identity adds noise to the entries of the count vector by the Laplace mechanism \cite{dwork2006differential} and cannot directly perturb high-dimensional datasets due to the domains being too large; we measure estimated error using the method described in \cite{mckenna2018optimizing}.}

Our target use case is privacy-preserving data exploration on high-dimensional data, for which the p-view should have the following four properties:
\begin{itemize}
    \item \textbf{Workload independence}: Data analysts desire to issue arbitrary queries for exploring data. These queries should not be predefined.
    \item \textbf{Analytical reliability}: For reliability in practical data exploration, it is necessary to be able to estimate the scale of the error for arbitrary counting queries.
    \item \textbf{Noise resistance on high-dimensional data}: Range counting queries compute the sum over the range and accumulate the noise injected for DP. This noise accumulation makes the query answers useless. To avoid this issue, we need a robust mechanism even for high-dimensional data. 
    \item \textbf{Space efficiency}: It is necessary to generate spatially efficient views even for count tensors with a large number of total domains on various datasets.
\end{itemize}

% State-of-the-art methods assume either a given workload \cite{li2015matrix, mckenna2018optimizing, yuan2016convex, yuan2012low} (i.e., predefined set of queries) or low dimensionality of data \cite{kotsogiannis2019privatesql, li2014data, Zhang2014TowardsAH} to release query answers with smaller errors caused by the noise.
% The workload-aware optimization methods \cite{li2015matrix, mckenna2018optimizing} is one of the promising techniques to release query answers for high-dimensional data; however, it cannot provide query-independent p-view.
% Data-aware partitioning methods \cite{li2014data, Zhang2014TowardsAH, qardaji2013differentially, xiao2012dpcube, yaroslavtsev2013accurate} are potential solutions, but they only focus on the low-dimensional data due to high complexity of discovering the optimal solution when the data has multiple attributes.
% Privtree \cite{zhang2016privtree} provides heuristic block partitioning for multidimensional count tensors, but it is also only suitable for low-dimensional data because of the non-adaptive cutting method.
% It may cause a lot of unnecessary block partitioning on high-dimensional data, which degrades accuracy and space efficiency.
% This paper attempts to close the gap for actualizing p-view having the above properties.
% Table \ref{tbl:salesman} summarizes the comparison between \method and these methods, and each method are described in Section \ref{sec:related_works}.

\noindent
\textbf{Our proposal.}
%To solve the above issues, we propose an recursive bisected partitioning method, \method.
To satisfy all the above requirements, we propose a simple yet effective recursive bisection method on a high-dimensional count tensor, \method.
Our proposed method has the same principle as \cite{xiao2012dpcube, li2014data, zhang2016privtree} of first partitioning a database into small blocks and then averaging over each block with noise.
Unlike the existing methods, \method can efficiently perform error-minimizing partitioning even on multidimensional count tensors instead of conventional 1D count vectors.
\method recursively partitions multidimensional blocks at a cutting point chosen in a differentially private manner while aiming to minimize the sum of aggregation error (AE) and perturbation error (PE).
Compared to Privtree \cite{zhang2016privtree}, our proposed method provides a more data-aware flexible cutting strategy and proper convergence of block partitioning, which results in smaller errors in private counting queries and much better spatial efficiency of the generated p-views.
Our method provides a powerful and practical solution for constructing a p-view under the DP constraint.
More importantly, the p-view generated by \method can work as a query processing system and expose the estimated error bound at runtime for any counting query without further privacy consumption.
This error information ensures reliable analysis for data explorers.
% An upper bound on the distribution of errors can be computed just from the published partitioning information.
% This is achieved by keeping simplicity of the whole randomization mechanism.
% Such error estimations can also be done with the query-optimization methods \cite{mckenna2018optimizing, li2015matrix}, but in that case, the query must be statically predefined.

\begin{table*}[t]
    \centering
    \caption{\method provides low-error counting queries in average on various workloads and datasets, and high space-efficiency of privacy-preserving materialized view (p-view) when $\epsilon = 1.0$. (N/A is due to HDMM and PrivBayes do not create p-view.)}
    \begin{tabular}{lrrrrr}
    \toprule
    & Identity\footnotemark[1] \cite{dwork2006differential} & Privtree \cite{zhang2016privtree} & HDMM \cite{mckenna2018optimizing} & PrivBayes \cite{zhang2017privbayes} & \textbf{\method (ours)} \\
    \midrule
    Average relative RMSE
    & $1.94\times 10^7$
    & $7.05$
    & $35.34$
    & $3.79$
    & $\mathbf{1.00}$
    \\
    Average relative size of p-view
    & $4.59\times 10^{17}$
    & 5578.27
    & N/A
    & N/A
    & $\mathbf{1.00}$
    \\
    \bottomrule
    \end{tabular}
\label{tbl:query_arr}
\end{table*}

\noindent
\textbf{Contributions.}
Our contributions are threefold.
First, we design a \textit{p-view} and formalize the segmentation for a multidimensional count tensor to find an effective p-view as error minimizing optimization problem.
P-view can be widely used for data exploration process on multidimensional data and is a differentially private approximation of a multidimensional histogram that can release counting queries with analytical reliability.
Second, we propose \method described above to find a desirable solution to the optimization problem.
Our algorithm is more effective than conventional algorithms due to finding flexible partitions and more efficient due to making appropriate convergence decisions.
Third, we conduct extensive experiments, whose source code and dataset are open, and show that \method has the following merits.
(1) Effectiveness: \method demonstrates smaller errors for various range counting queries and outperforms the existing methods \cite{dwork2006differential, zhang2016privtree, li2014data, mckenna2018optimizing, zhang2017privbayes} on multi-dimensional real-world datasets. 
% Furthermore, our method also shows high performance in classification tasks on sampling synthesized data from the p-view.
% (2) Space efficieny: \method generates a compact representation of the p-view, and it is $5\times10^{3}$ times smaller on average than that of the state-of-the-art data-aware segmentation method \cite{zhang2016privtree} in our experiment.
(2) Space efficiency: \method generates a much more compact representation of the p-view than the state-of-the-art (i.e., Privtree \cite{zhang2016privtree}) in our experiment.
% (3) Scalability: the run time of \method is sublinear with respect to the domain size of the data.

\noindent
\textbf{Preview of result.}
We present a summary previewing of the experimental results.
Table \ref{tbl:query_arr} shows the average relative root mean squared error against (RMSE) of \method in eight types of range counting queries on eight real-world datasets and the average relative size of the p-view generated by the algorithms. 
With Identity, we obtain a p-view by making each cell of the original count tensor a converged block.
% \footnote{Note, the scores of DAWA and Identity are only on low-dimensional datasets.}
\method yields the smallest error score on average.
% Several algorithms show better results than \method in some queries, but \method is the best on average.
This is a desirable property for data explorations.
Furthermore, compared to that of Privtree \cite{zhang2016privtree}, the p-view generated by \method is more space efficient.
% This indicates that the \method can perform effective convergence decisions in addition to finding effective block partitioning.
% The evaluation part of this paper reports extensive experiments on range counting queries, and additionally, to evaluate \method in the random sampling use case, we conduct classifications using the dataset sampled from p-view on five real datasets.

% We believe that our proposed method could help data analysts explore sensitive data in the early stages of data mining pipelines while preserving data utility and privacy.
% Therefore, in a practical data science workflow, our proposed method is useful to design the workload in-house, before applying the workload-aware query processing to release the private results.
% Furthermore, the advantage that our method can release an unlimited number of arbitrary queries without additional privacy budget is useful for providing an interface for interactive queries to third parties.
% Most of existing DP-query engines \cite{johnson2018towards, johnson2018chorus, ge2019apex} consider privacy consumption by sequential composition or advanced sub-linear composition \cite{rogers2020linkedin}, which results in restricting data explorations.

% In a data science workflow, our method is a good fit the in-house data exploration, and a workload-aware method is suitable for releasing query answers to third parties.
% This is a reasonable scenario in a data science workflow.
% This paper focuses on the data exploration phase, and a study about the whole workflow is future work.

\section{Related Works}
\label{sec:related_works}
In the last decade, several works have proposed differentially private methods for exploring sensitive data.
Here, we describe the state-of-the-arts related to our work.
% Table \ref{tbl:salesman} summarizes it.
% In this section, we review some of the methods that are particularly relevant to our study.

\noindent
\textbf{Data-aware partitioning.}
Data-aware partitioning is a conventional method that aims to directly randomize and expose the entire histogram for all domains (e.g., count vector, count tensor); thereby, it can immediately compose a p-view that answers all counting queries.
A na\"ive approach to constructing a differentially private view is adding Laplace noise \cite{dwork2006differential} to all values of a count vector; this is called the Identity mechanism.
This na\"ive approach results in prohibitive noise on query answers through the accumulation of noise over the grouped bins used by queries.
DAWA \cite{li2014data} and AHP \cite{Zhang2014TowardsAH} take data-aware partitioning approaches to reduce the amount of noise.
The partitioning-based approaches first split a raw count vector into bins and then craft differentially private aggregates by averaging each bin and injecting a single unit of noise in each bin.
However, these approaches work only for very low (e.g., one or two) -dimensional data due to the high complexity of discovering the optimal solution when the data have multiple attributes.
DPCube \cite{xiao2012dpcube} is a two-step multidimensional data-aware partitioning method, but the first step, obtaining an accurate approximate histogram, is difficult on high-dimensional data with small counts in each cell.

Privtree \cite{zhang2016privtree} and \cite{cormode2012differentially} perform multidimensional data-aware partitioning on count tensors, mainly targeting the spatial decomposition task for spatial data.
Unlike our method, this method uses a fixed quadtree as the block partitioning strategy, which leads to an increase in unnecessary block partitioning as the dimensionality increases. 
As a result, it downgrades the spatial efficiency and incurs larger perturbation noise.
In addition, this method aims to partition the blocks such that the count value is below a certain threshold, while our proposed method aims to minimize the AE of the blocks and reduce count query noise.

\noindent
\textbf{Optimization of given workloads.}
Another well-established approach is the optimization for a given workload.
Li et al. \cite{li2015matrix} introduced a matrix mechanism (MM) that crafts queries and outputs optimized for a given workload.
The high-dimensional MM (HDMM) \cite{mckenna2018optimizing} is a workload-aware data processing method extending the MM to be robust against noise for high-dimensional data.
PrivateSQL \cite{kotsogiannis2019privatesql} selects the view to optimize from pregiven workloads.
In the data exploration process, it is not practical to assume a predefined workload, and these methods are characterized by a loss of accuracy when optimized for a workload of wide variety of queries.
% However, for exploring the sensitive data, existing differentially private views and query processing methods requires either a given workload or low dimensionality for the data.
% That means the existing approaches result in high aggregation errors for the multi-dimensional data without the given workload.

\noindent
\textbf{Private data synthesis.}
Private data synthesis, which builds a privacy-preserving generative model of sensitive data and generates synthetic records from the model, is also useful for data exploration.
Note that synthesized dataset can work as a p-view by itself.
PrivBayes \cite{zhang2017privbayes} can heuristically learn a Bayesian network of data distribution in a differentially private manner.
DPPro \cite{8006304}, Priview \cite{qardaji2014priview} and PrivSyn \cite{Zhang2021PrivSynDP} represent distribution by approximation with several smaller marginal tables.
While these methods provide a partial utility guarantee based on randomized mechanisms such as Laplace mechanisms or random projections, they face difficulties in providing an error bound for arbitrary counting queries on the synthesized data.
Differentially private deep generative models have also attracted attention \cite{acs2018differentially, jordon2018pate, takagi2021p3gm, harder2021dp}, but most of the works focus on the reconstruction of image datasets.
Fan et al. \cite{fan2020relational} studied how to build a good generative model based on generative adversarial nets (GAN) for tabular datasets.
Their experimental results showed that the utility of differentially private GAN was lower than that of PrivBayes for tabular data.
% P3GM \cite{takagi2021p3gm} and DP-MERF \cite{harder2021dp} introduced noise resistant learning algorithms for high-dimensional image data, but did not achieve sufficient performance on tabular data.
\cite{ren2018textsf} provides a solution for high-dimensional data in a local DP setting.

As mentioned in Section \ref{sec:intro}, the accuracy of these methods has improved greatly in recent years, but it is difficult to guarantee their utility for analysis using counting queries, and there are large gaps in practice.
DPPro \cite{8006304} utilizes a random projection \cite{johnson1984extensions} that preserve L2-distance to the original data in an analyzable form to give a utility guarantee, but this is different from the guarantee for each counting query.
CSM \cite{zheng2019differentially} gives a utility analysis for queries, however, their analysis ignores the effect of information loss due to compression, which may not be accurate.
Also, as shown in their experiments, they apply intense preprocessing to the domain size and do not show the effectiveness for high-dimensional data.
Our proposed method provides an end-to-end error analysis for arbitrary counting queries by directly constructing p-views from histograms without any intermediate generative model.

\noindent
\textbf{Querying and programming framework.}
PrivateSQL \cite{kotsogiannis2019privatesql} is a differentially private relational database system for multirelational tables, where for each table, it applies an existing noise-reducing method such as DAWA.
Unlike our method, PrivateSQL needs a given workload to design private views to release.
Flex \cite{johnson2018towards}, Google's work \cite{wilson2019differentially} and APEx \cite{ge2019apex} are SQL processing frameworks under DP.
They issue queries to the raw data, which can consume an infinite amount of the privacy budgets.
Hence, we believe that these DP-query processing engines are not suitable for data exploration tasks where many instances of trial and error may be possible.
% Crypt$\epsilon$ \cite{roy2020crypt} introduced a query processing with a hybrid privacy model of differential privacy and homomorphic encryption.
% PINQ \cite{mcsherry2009privacy} and $\epsilon$ktelo \cite{zhang2018ektelo} gives a programming framework for differentially private algorithms.
% Our proposed method is also applicable in these frameworks.
% DPBench \cite{hay2016principled} summarizes the performance comparisons of differentially private mechanisms.
Our method generates p-view, which can be used as a differentially private query system that allows any number of range counting queries to be issued.
% In our method, we perform differential private query processing by issuing count range queries on the generated p-views.

%DPBench \cite{hay2016principled} summarizes the performance comparisons of differentially private mechanisms.

\section{Preliminaries}

% \begin{table}[t]
%     \centering
%     \begin{tabular}{cl}
%     \toprule
%     Symbol & \multicolumn{1}{c}{Definition} \\
%     \midrule
%     $x$ & \\
%     $X$ & hoge \\
%     $\mathbf{A}$ & attribute set \\
%     $dom(a)$ & domains \\
%     $\mathcal{X}$ & count tensor \\
%     $\mathcal{X'}$ & p-view, sanitized count tensor \\
%     $\mathcal{B}$ & block \\
%     $W$ & workload \\
%     $D, D'$ & neighboring database \\
%     $AE(\mathcal{B})$ & aggregation error of block $\mathcal{B}$ \\
%     \bottomrule
%     \end{tabular}
%     \caption{Table of notations. WIP}
%     \label{tab:notations}
% \end{table}

This section introduces essential knowledge for understanding our proposal.
We first describe notations this paper uses.
Then, we briefly explain DP.

% Table \ref{tab:notations} summarizes notations used in this paper.

\subsection{Notation}

Let $X$ be the input database with $n$ records consisting of an attribute set $A$ that has $d$ attributes $A=\{a_1, \dots, a_d\}$.
The domain $dom(a)$ of an attribute $a$ has a finite ordered set of discrete values, and the size of the domain is denoted as $|dom(a)|$.
The overall domain size of $A$ is $|dom(A)|=\prod_{i\in[d]} |dom(a_i)|$, where $[d]=\{1,\dots,d\}$.
In the case where attribute $a$ is continuous, we transform the domain into a discrete domain by binning, and in the case where attribute $a$ is categorical, we transform it into an ordered domain.
Then, $dom(a)$ can be represented as a \textit{range} $r[s_a, e_a]$ where for all $p_a \in dom(a)$, $s_a \le p_a \le e_a$.
For ranges $r_1, r_2$, $|r_1 \cap r_2|$ means the number of value $p_a$ satisfies $s_a \leq p_a \leq e_a $.

We consider transforming the database $X$ into the $d$-mode count tensor $\mathcal{X}$, where
% , where $\mathcal{X}[i_1,\dots,i_d]$ represents the number of records where $(a_1=r[i_1, i_1], \dots, a_d=r[i_d, i_d]) \in X$.
given $d$ ranges $r_1,\dots,r_d$, 
$\mathcal{X}[r_1,\dots,r_d]$ represents the number of records where $(a_1 (\in r_1), \dots, a_d (\in r_d)) \in X$.
We utilize $x$ $(\in \mathcal{X})$ as a count value in $\mathcal{X}$; this corresponds to a cell of the count tensor.
% Thus, the count tensor can answer all predicate counting queries.
We denote a subtensor of $\mathcal{X}$ as \textit{block} $\mathcal{B} \subseteq \mathcal{X}$.
$\mathcal{B}$ is also a $d$-mode count tensor, but its domain in each dimension is smaller than or equal to that of the original count tensor $\mathcal{X}$; i.e., each attribute $a_i$ $(i\in [d])$, $r[s_{a_i}, e_{a_i}]$ of $\mathcal{B}$ and $r[s'_{a_i}, e'_{a_i}]$ of $\mathcal{X}$ satisfy $s'_{a_i} \le s_{a_i}$ and $e_{a_i} \le e'_{a_i}$.
We denote the domain size of $\mathcal{B}$ as $|\mathcal{B}|$.

Last, we denote $q$ as a counting query and $\mathbf{W}$ as a workload.
$\mathbf{W}$ is a set of $|\mathbf{W}|$ counting queries, where $\mathbf{W} = \{q_1, ..., q_{|\mathbf{W}|}\}$, and $q(\mathcal{X})$ returns the counting query results for count tensor $\mathcal{X}$.

\subsection{Differential Privacy}
DP \cite{dwork2006differential} is a rigorous mathematical privacy definition that quantitatively evaluates the degree of privacy protection when we publish outputs.
DP is used in broad domains and applications \cite{bindschaedler2017plausible, chaudhuri2019capacity, papernot2016semi}. 
The importance of DP is supported by the fact that the US census announced '2020 Census results will be protected using “differential privacy”, the new gold standard in data privacy protection' \cite{abowd2018us}. 

\begin{definition}[$\epsilon$-differential privacy]
A randomized mechanism $\mathcal{M}:\mathcal{D}\rightarrow\mathcal{Z}$ satisfies $\epsilon$-DP if, for any two inputs $D, D' \in \mathcal{D}$ such that $D'$ differs from $D$ in at most one record and any subset of outputs $Z \subseteq \mathcal{Z}$, it holds that
\begin{equation}
\nonumber
  \Pr[\mathcal{M}(D)\in Z] \leq \exp(\epsilon) \Pr[\mathcal{M}(D')\in Z].
\end{equation}
\end{definition}
\noindent
We define databases $D$ and $D'$ as \textit{neighboring} databases.

Practically, we employ a randomized mechanism $\mathcal{M}$ that ensures DP for a function $f$.
The mechanism $\mathcal{M}$ perturbs the output of $f$ to cover $f$'s sensitivity, which is the maximum degree of change over any pair of datasets $D$ and $D'$.

\begin{definition}[Sensitivity]
The sensitivity of a function $f$ for any two neighboring inputs $D, D' \in \mathcal{D}$ is:
\begin{equation}
\nonumber
    \Delta_{f} = \sup_{D, D' \in \mathcal{D}} \|f(D)-f(D')\|.
\end{equation}
where $||\cdot||$ is a norm function defined in $f$'s output domain.
\end{definition}
\noindent
When $f$ is a histogram, $\Delta_{f}$ equals 1 \cite{hay2009boosting}.
Based on the sensitivity of $f$, we design the degree of noise to ensure DP.
The Laplace mechanism and exponential mechanism are well-known as standard approaches.

The Laplace mechanism can be used for randomizing numerical data.
Releasing the histogram is a typical use case of this mechanism.
\begin{definition}[Laplace Mechanism]
For function $f:\mathcal{D}\rightarrow\mathbb{R}^{n}$, the Laplace mechanism adds noise $f(D)$ as:
\begin{equation}
    f(D) + \text{Lap}(\Delta_f / \epsilon)^{n} .
\end{equation}
where $\text{Lap}(\lambda)^{n}$ denotes a vector of $n$ independent samples from a Laplace distribution $Lap(\lambda)$ with mean 0 and scale $\lambda$.
\end{definition}

The exponential mechanism is the random selection algorithm.
The selection probability is weighted based on a score in a quality metric for each item.
\begin{definition}[Exponential Mechanism]
Let $q$ be the quality metric for choosing an item $y \in Y$ in the database $D$.
The exponential mechanism randomly samples $y$ from $Y$ with weighted sampling probability defined as follows:
\begin{equation}
    \Pr[y] \sim \exp(\frac{\epsilon q(D, y)}{2\Delta_q}) .
\end{equation}
\end{definition}

Quantifying the privacy of differentially private mechanisms is essential for releasing multiple outputs.
Sequential composition and parallel composition are standard privacy accounting methods.

\begin{theorem}[Sequential Composition \cite{dwork2006differential}]
Let $\mathcal{M}_1, \dots, \mathcal{M}_k$ be mechanisms satisfying $\epsilon_1$-, $\dots, \epsilon_k$-DP.
Then, a mechanism sequentially applying $\mathcal{M}_1, \dots, \mathcal{M}_k$ satisfies ($\sum_{i\in[k]} \epsilon_i$)-DP.
\end{theorem}

\begin{theorem}[Parallel Composition \cite{mcsherry2009privacy}]
Let $\mathcal{M}_1, \dots, \mathcal{M}_k$ be mechanisms satisfying $\epsilon_1$-, $\dots, \epsilon_k$-DP.
Then, a mechanism applying $\mathcal{M}_1, \dots, \mathcal{M}_k$ to the disjoint datasets $D_1, \dots , D_k$ in parallel satisfies $(\max_{i\in [k]}\epsilon_i)$-DP.
\end{theorem}

\section{Problem Formulation}
\label{sec:problem_formulation}

\subsection{Segmentation as Optimization}

\begin{figure*}
    \centering
    \includegraphics[width=0.95\hsize]{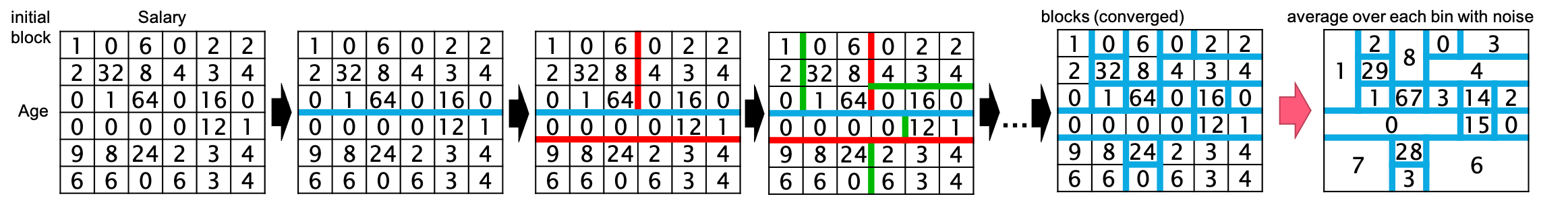}
    \caption{\method efficiently discovers blocks (i.e., groups of count cells) with smaller AEs (black arrow) and averages over each block with injected noise (red arrow). The p-view stores the randomized counts in a blockwise way.}
    \label{fig:algo_overview}
\end{figure*}

This section describes the foundation of multidimensional data-aware segmentation that seeks a solution for the differentially private view $\Tilde{\mathcal{X}}$ from the input count tensor $\mathcal{X}$.
Every count $\Tilde{x} \in \Tilde{\mathcal{X}}$ is sanitized to satisfy DP.
We formulate multidimensional block segmentation as an optimization problem.

\noindent
\textbf{Foundation.}
Given a count tensor $\mathcal{X}$, we consider partitioning $\mathcal{X}$ into $m$ blocks $\pi=\{\mathcal{B}_1, ... ,\mathcal{B}_{m}\}$.
The blocks satisfy $\mathcal{B}_i \cap \mathcal{B}_j=\emptyset$ where $i,j\in [m]$, $j\neq i$ and $\mathcal{B}_1 \cup \dots \cup \mathcal{B}_m=\mathcal{X}$.
We denote the sum over $\mathcal{B}_i$ as $S_i=\sum_{x'\in \mathcal{B}_i} x'$ and its perturbed output as $\Tilde{S_i}=S_i + z_i$.
We can sample $z_i$ with the Laplace mechanism $Lap(1/\epsilon)$ and craft the $\epsilon$-differentially private sum in $\mathcal{B}_i$.

For any count $x$ in the block $\mathcal{B}_i$, we have two types of errors: 
\textit{Perturbation Error} (PE) and \textit{Aggregation Error} (AE).
Assuming that we replace any count $x \in \mathcal{B}_i$ with $\bar{x}_i=(S_i+z_i)/|\mathcal{B}_i|$, the absolute error between $x$ and $\bar{x}_i$ can be computed as
\begin{equation}
    |x-\bar{x}_i|  
    = \left|\left(x-\frac{S_i}{|\mathcal{B}_i|}\right)-\frac{z_i}{|\mathcal{B}_i|}\right|
    \leq \left|x-\frac{S_i}{|\mathcal{B}_i|}\right|+\left|\frac{z_i}{|\mathcal{B}_i|}\right| .
\end{equation}
Therefore, the total error over block $\mathcal{B}_i$, namely, the segmentation error (SE),  can be given by:
\begin{equation}
    \text{SE}(\mathcal{B}_i) = \sum_{x\in \mathcal{B}_i} |x-\bar{x}_i| \leq \text{AE}(\mathcal{B}_i) + \text{PE}(\mathcal{B}_i)
    \label{eq:two_errors}
\end{equation}
where
\begin{align}
    \text{AE}(\mathcal{B}_i) &\coloneqq \sum_{x\in \mathcal{B}_i} \left|x-\frac{S_i}{|\mathcal{B}_i|}\right|, 
    \label{eq:ae} \\
    \text{PE}(\mathcal{B}_i) &\coloneqq \left|z_i\right| . 
    \label{eq:pe}
\end{align}
(\ref{eq:ae}) and (\ref{eq:pe}) represent the AE and the PE, respectively.

\noindent
\textbf{Problem.}
The partitioning makes the PE of each block $\frac{1}{|\mathcal{B}_i|}$ times smaller than those of the original counts with Laplace noise.
Furthermore, we consider the expectation of the SE
\begin{equation}
    \small
    \begin{split}
        \mathbb{E}\left[\sum_{i \in [m]} \text{SE}(\mathcal{B}_i)\right]
            &\leq \mathbb{E}\left[\sum_{i \in [m]} \text{AE}(\mathcal{B}_i)\right] + \mathbb{E}\left[\sum_{i \in [m]} \text{PE}(\mathcal{B}_i)\right]\\
            &= \sum_{i \in [m]} \text{AE}(\mathcal{B}_i) + \sum_{i \in [m]} \mathbb{E}\left[\text{PE}(\mathcal{B}_i)\right] \\
            &= \sum_{i \in [m]} \text{AE}(\mathcal{B}_i) + m \cdot \frac{1}{\epsilon} .
    \end{split}
    \label{eq:expect_error}
\end{equation}
Thus, to discover the optimal partition $\pi$, we need to minimize Eq. (\ref{eq:expect_error}).
The optimization problem is denoted as follows:
\begin{equation}
\begin{split}
\underset{\mathbf{\pi}} {\text{minimize}} & \sum_{\mathcal{B}_i \in \pi} \left(
\text{AE}(\mathcal{B}_i) + \frac{1}{\epsilon}\right) \\
\text{subject to} & \ \mathcal{B}_i \cap \mathcal{B}_{j\neq i} = \emptyset,\;\; \mathcal{B}_i, \mathcal{B}_j \in \pi \\ 
& \bigcup_{\mathcal{B}_i \in \pi}{\mathcal{B}_i} = \mathcal{X}
\end{split}
\label{eq:optim}
\end{equation}

\noindent
\textbf{Challenges.} 
It is not easy to discover the optimal partition $\pi$.
This problem is an instance of the set partitioning problem \cite{balas1976set}, which is known to be NP-complete, where the objective function is computed by brute-force searching for every combination of candidate blocks.
It is hard to solve since the search space is basically a very large scale due to large $|dom(A)|$.
Therefore, this paper seeks an efficient heuristic solution with a good balance between utility (i.e., smaller errors) and privacy.

\subsection{P-view Definition}
Our proposed p-view has a simple structure.
The p-view consists of a set of blocks, each of which has a range for each attribute and an appropriately randomized count value, as shown in Figure \ref{fig:overview}.
Formally, we define the p-view as follows.
\begin{equation}
\label{eq:p_view}
\centering
\begin{split}
    \mbox{p-view}\; \Tilde{\mathcal{X}} &= \{ \mathcal{B}_1,...,\mathcal{B}_m \},  \\
    \mbox{for } i\in [m],  \mathcal{B}_i &= (\{r[s^{(i)}_{a_1}, e^{(i)}_{a_1}],...,r[s^{(i)}_{a_d}, e^{(i)}_{a_d}]\}, \Tilde{S_i})
\end{split}
\end{equation}
Thus, each block $\mathcal{B}_i$ has this $d$-dimensional domain and the sanitized sum of count values $\Tilde{S_i}$.

In the range counting query processing, a counting query $q$ needs to have the range condition $c_q = \{r[s^{(q)}_{a_1}, e^{(q)}_{a_1}],...,r[s^{(q)}_{a_d}, e^{(q)}_{a_d}]\}$.
Let the ranges of $\mathcal{B}_i$ be $\{r[s^{(i)}_{a_1}, e^{(i)}_{a_1}],...,r[s^{(i)}_{a_d}, e^{(i)}_{a_d}]\}$, and we calculate the intersection of $c_q$ and the block and add the count value according to the size of the intersection.
Hence, the result can be calculated as follows.
\begin{equation}
    \label{eq:count_query}
    \small
    q(\mathcal{\Tilde{\mathcal{X}}}) = \sum_{i=1,...,m}{\left(\prod_{l=1,...,d}{\left(\left|r[s^{(q)}_{a_l}, e^{(q)}_{a_l}] \cap r[s^{(i)}_{a_l}, e^{(i)}_{a_l}]\right|\right)}*\frac{\Tilde{S_i}}{|\mathcal{B}_i|}\right)}
\end{equation}
% where $|c_q \cap r_i|$ is $\prod_{l=1,...,d}{(r[s^{(q)}_{a_l}, e^{(q)}_{a_l}] \cap r[s^{(i)}_{a_l}, e^{(i)}_{a_l}])}$.
The number of intersection calculations is proportional to the number of blocks, and the complexity of the query processing is $\mathcal{O}(md)$.
% ; in general, the lower the number of blocks is, the faster the query processing is.
% Practically, we reduce the computation by using the block division path as an tree-index to filter the blocks to be checked.

% (remove?) Additionally, p-view allows private data sampling, which is also essential in data exploration when synthetic data is needed.
% Since the p-view keeps the domain range of the data and its frequency for each block, it is an approximated histogram of the original data such as the last figure of Figure \ref{fig:algo_overview}.
% Therefore, by normalizing the count value of each block, we obtain the distribution of the appearance probability of each block $\mathcal{B}_i$ as follows.
% \begin{equation}
%     \Pr[\mathcal{B}_i] = \frac{\Tilde{S_i}}{\sum{\Tilde{S_j}}}
% \end{equation}
% Further, in sampling, it is necessary to sample a data, not a block, so after sampling the block, one domain $(a_1,...,a_d)$ is uniformly sampled from the domain of the block.

\section{Proposed Algorithm}

This section introduces our proposed solution.
Our solution constructs a p-view of the input relational data while preserving utility and privacy with analytical reliability to estimate errors in the arbitrary counting queries against the p-view (Eq.(\ref{eq:count_query})).

% , reducing both AEs for segmentation and PEs for DP.

\subsection{Overview}

Our challenge is to devise a simple yet effective algorithm that enables us to efficiently search a block partitioning with small total errors and DP guarantees.
As a realization of the algorithm, we propose \method.
% Moreover, it should be possible to estimate errors in the arbitrary counting queries against the p-view (Eq.(\ref{eq:count_query})).
% As a realization of the algorithm, we propose \method.
% \method employs the recursive bisection framework that the fast and scalable k-anonymization algorithm Mondrian \cite{lefevre2006mondrian} introduced.
% Although our algorithm recursively bisects a block as well as the Mondrian, we also have the challenge of how to organize well blocks that reduce the aggregation errors within the limitation of a given privacy budget.

Figure \ref{fig:algo_overview} illustrates an overview of our proposed algorithm.
First, \method creates the initial block $\mathcal{B}^{(0)}$ that covers the whole count tensor $\mathcal{X}$.
Second, we recursively bisect a block $\mathcal{B}$ (initially $\mathcal{B}=\mathcal{B}^{(0)}$) into two disjoint blocks $\mathcal{B}_L$ and $\mathcal{B}_R$. %($\mathcal{B}_L\cup \mathcal{B}_R=\mathcal{B}$ and $\mathcal{B}_L\cap \mathcal{B}_R=\emptyset$).
% Before bisecting $\mathcal{B}$ into $\mathcal{B}_L$ and $\mathcal{B}_R$, we check whether the AE over $\mathcal{B}$ is sufficiently small. 
% If the result of the check is positive, we stop the recursive bisection for $\mathcal{B}$.
Before bisecting $\mathcal{B}$, we check whether the AE over $\mathcal{B}$ is sufficiently small. 
If the result of the check is positive, we stop the recursive bisection for $\mathcal{B}$.
Otherwise, we continue to split $\mathcal{B}$.
We pick a splitting point $p \in dom(a)$ ($a \in A$) for splitting $\mathcal{B}$ into $\mathcal{B}_L$ and $\mathcal{B}_R$ which have smaller AEs.
Although splitting does not always result in smaller total AEs, proper cut point obviously makes AEs much smaller.
Third, \method recursively executes these steps separately for $\mathcal{B}_L$ and $\mathcal{B}_R$.
After convergence is met for all blocks, \method generates a randomized aggregate by $S_i+z_i$ where $z_i\sim Lap(1/\epsilon)$ for each block $\mathcal{B}_i$.
Finally, for all $x \in \mathcal{B}_i$, we obtain the randomized count $\Tilde{x}=(S_i+z_i)/|\mathcal{B}_i|$.

% The abovementioned algorithm can discover blocks that heuristically reduce the AEs, and is efficient due to its simplicity.
% The necessary space complexity is linear with respect to the number of records, and the total depth of splitting is sublinear with respect to domain size.
% However, the question is \textit{how can we make the above algorithm differentially private?}.

The abovementioned algorithm can discover blocks that heuristically reduce the AEs, and is efficient due to its simplicity.
However, the question is \textit{how can we make the above algorithm differentially private?}.
To solve this question, we introduce two mechanisms, \textit{random converge} (Section \ref{sec:random_converge}) and \textit{random cut} (Section \ref{sec:random_cut}).
Random converge determines the convergence of the recursive bisection, and random cut determines the effective single cutting point.
These provide reasonable partitioning strategy to reduce the total errors with small privacy budget consumption.

% We explain them in detail in the next sections.
% Our proposed algorithm with these two mechanisms can discover reasonable cutting strategy with small privacy budget consumption, and hence it can dig into deep bisections to reduce the AEs while saving the total privacy budget of \method.

The overall algorithm of \method are described in Algorithm \ref{algo:proposed}.
Let $\epsilon_b=\epsilon_r+\epsilon_p$ be the total privacy budget for \method, where $\epsilon_r$ is the budget for the recursive bisection and $\epsilon_p$ is the budget for the perturbation.
\method utilizes $\gamma \epsilon_r$ for random converge and $(1-\gamma) \epsilon_r$ for random cut ($0 \le \gamma \le 1$).
$\alpha$ is a hyperparameter that determines the size of $\lambda$ and $\delta$, where $\lambda$ corresponds to the Laplace noise scale of random converge (Lines \ref{line:lambda}, \ref{line:judge}) and $\delta$ is a bias term for AE (Lines \ref{line:delta}, \ref{line:bae}).
These are, sketchily, tricks for performing random converge with depth-independent scales, which are explained in Section \ref{sec:random_converge} and a detailed proof of DP is given in Section \ref{sec:privacy_accouting}.
The algorithm runs recursively (Lines \ref{line:run_1}, \ref{line:run_2}, \ref{line:run_3}), alternating between random converge (Lines \ref{line:conv_s}-\ref{line:conv_e}) and random cut (Lines \ref{line:cut_s}-\ref{line:cut_e}).
The random converge stops when the AE becomes small enough, consuming a total budget of $\gamma \epsilon_r$ independent of the number of depth.
The random cut consumes a budget of $(1-\gamma) \epsilon_r /\kappa$ for each cutting point selection until the depth exceeds $\kappa$.
$\kappa$ is set as $\kappa=\beta\log_2{\bar{n}}$, where $\beta > 0$ is hyperparameter and $\bar{n}$ is the total domain size of the data.
As we see later in Theorem \ref{theorem:monotonic_ae}, AE is not increased by splitting, so if the depth is greater than $\kappa$, we split randomly without any privacy consumption until convergence.
After the recursive bisection converges, \method perturbs the count by adding the Laplace noise while consuming $\epsilon_p$.

\begin{algorithm}
    \caption{\method}\label{algo:proposed}
    \algrenewcommand\algorithmicrequire{\textbf{Input:}}
    \algrenewcommand\algorithmicensure{\textbf{Output:}}
    \begin{algorithmic}[1]
    \Require initial block $\mathcal{B}^{(0)}$, privacy budget $\epsilon_{b}$, recursive bisection budget ratio $\epsilon_{r}/\epsilon_{b}$, hyperparameters $\alpha$, $\beta$, $\gamma$
    \Ensure p-view $\Tilde{\mathcal{X}}$
    \Procedure{\method}{$\mathcal{B}^{(0)}, \epsilon_{b}, \epsilon_{r}/\epsilon_{b}, \alpha, \beta, \gamma$}
        \State $\epsilon_{r} \leftarrow \epsilon_{b} \cdot (\epsilon_{r}/\epsilon_{b})$; \quad $\epsilon_{p} \leftarrow \epsilon_{b} \cdot (1-\epsilon_{r}/\epsilon_{b})$
        \State $\Tilde{n} \leftarrow$ \textsc{TotalDomainSizeOf}$(\mathcal{X})$
        \State $\kappa \leftarrow \beta \log_2{\Tilde{n}}$ \quad // maximum depth of random cut
        \State $\pi \leftarrow \{\}$; \quad $k \leftarrow 1$ \quad // converged blocks; current depth
        \State $\theta \leftarrow 1/\epsilon_p$ \quad // threshold
        \State $\epsilon_{cut} \leftarrow (1-\gamma) \epsilon_r / \kappa$ \quad // privacy budget for random cut
        \State $\lambda \leftarrow \left(\frac{2\alpha-1}{\alpha-1} + 1\right) \cdot \left(\frac{2}{\gamma \epsilon_r}\right)$ // noise scale for random converge\label{line:lambda}
        \State $\delta \leftarrow \lambda \log{\alpha}$ \quad // bias parameter \label{line:delta}
        \State \textsc{RecursiveBisection}$(\mathcal{B}^{(0)}, \pi, \epsilon_{cut}, k, \kappa, \theta, \lambda, \delta)$ \label{line:run_1}
        
        \State $\Tilde{\mathcal{X}} \leftarrow$ \textsc{Perturbation}$(\pi, \epsilon_{p})$
        \State \Return $\Tilde{\mathcal{X}}$
    \EndProcedure
    
    \Procedure{\textsc{RecursiveBisection}}{$\mathcal{B}$, $\pi$, $\epsilon_{cut}$, $k$, $\kappa$, $\theta$, $\lambda$, $\delta$}
        \State /* Random Converge */ \label{line:conv_s}
        \State $\text{BAE}(\mathcal{B}) \leftarrow \mbox{max}(\theta + 2 - \delta,\,  \text{AE}(\mathcal{B})-k\delta )$ \label{line:bae}
        \If {$\text{BAE}(\mathcal{B}) + \mathrm{Lap}(\lambda) \le \theta$} \label{line:judge}
        \State $\pi \leftarrow \pi \bigcup \mathcal{B}$
        \State \Return \label{line:conv_e}
        \EndIf 
        \State /* Random Cut */ \label{line:cut_s}
        \If {$k \le \kappa$}
        \ForAll{$i \in [d], j \in [|dom(a_i)|]$}
        \State $quality[i, j] \leftarrow Q(B,a_{ij})$
        \EndFor
        \State $(i^*,j^*)\leftarrow \textsc{WeightedSampling}(\epsilon_{cut}, quality)$
        \Else
        \State $(i^*,j^*)\leftarrow \textsc{RandomSampling}([d], [|dom(a_i)|])$ \label{line:random_sampling}
        \EndIf
        \State $(\mathcal{B}_{L}, \mathcal{B}_{R}) \leftarrow \textsc{Split}(i^*,j^*)$ \label{line:cut_e}
        \State /* Repeat Recursively */ \label{line:rec_s}
        \State \textsc{RecursiveBisection}$(\mathcal{B}_{L}, \pi, \epsilon_{cut}, k+1, \kappa, \theta, \lambda, \delta)$ \label{line:run_2}
        \State \textsc{RecursiveBisection}$(\mathcal{B}_{R}, \pi, \epsilon_{cut}, k+1, \kappa, \theta, \lambda, \delta)$ \label{line:run_3}
        \State \Return \label{line:rec_e}
    \EndProcedure
    \end{algorithmic}
\end{algorithm}

\subsection{Random Converge}
\label{sec:random_converge}
AE decreases by properly splitting the blocks, however unnecessary block splitting leads to an increase in PE as mentioned above.
To stop the recursive bisection at the appropriate depth, we need to obtain the exact AE of the block, which is a data-dependent output, therefore we need to provide a DP guarantee.
One approach is to publish differential private AE so that making the decision for the stop is also DP by the post-processing property.
In other words, the stop is determined by $\text{AE}(\mathcal{B}) + Lap(\lambda) \le \theta$ where $\theta$ is a threshold indicating AE is small enough.
However, this method consumes privacy budget every time the AE is published, and the budget cannot be allocated unless the depth of the partition is decided in advance.
Therefore, we utilize the observation for the \textit{privacy loss} of Laplace mechanism-based threshold query \cite{zhang2016privtree} and design the biased AE (BAE) of the block $\mathcal{B}$ instead of $\text{AE}(\mathcal{B})$: $\text{BAE}(\mathcal{B}) = \mbox{max}(\theta + 2 - \delta,\,  \text{AE}(\mathcal{B})-k\delta )$, where $k$ is the current depth of bisection, $\delta$ $(> 0)$ is a bias parameter, i.e., we determine the convergence by $\text{BAE}(\mathcal{B}) + Lap(\lambda) \le \theta$.
Intuitively, the BAE is designed to tightly bound the privacy loss of the any number of Laplace mechanism-based threshold queries with constant noise scale $\lambda$.
When the value is sufficiently larger than the threshold, this privacy loss decreases exponentially \cite{zhang2016privtree}.
Then, it can be easily bounded by an infinite series regardless of the number of queries.
Conversely, when the value is small compared to the threshold, each threshold query consumes a constant budget.
To limit the number of such budget consumptions, a bias $\delta$ is used to force a decrease in the value for each threshold query (i.e., each depth) because BAE has a minimum and if the value is guaranteed to be less than the minimum for adjacent databases, the privacy loss is zero.
The design of our BAE allows for two constant budget consumptions at most, with the remainder being bounded by an infinite series.
We give a detailed proof in Section \ref{sec:privacy_accouting}.
As a whole, since BAE is basically close to AE, AEs are expected to become sufficiently small overall.

Then, we consider about $\theta$ where if $\theta$ is too large, block partitioning will not sufficiently proceed, causing large AEs, and if it is too small, more blocks will be generated, leading to increase in total PEs.
To prevent unwanted splitting, it is appropriate to stop when the increase in PE is greater than the decrease in AE.
We design the threshold $\theta$ as $1/\epsilon_p$ which is the standard deviation of the Laplace noise to be perturbed.
Considering the each bisection increases the total PE by $1/\epsilon_p$, when the AE becomes less than the PE, the division will increase the error at least.
Hence, it is reasonable to stop under this condition.

\subsection{Random Cut}
\label{sec:random_cut}
Here, the primary question is how to pick a reasonable cutting point from all attribute values in a block $\mathcal{B}$ under DP.
Our intuition is that a good cutting point results in smaller AEs in the two split blocks.
We design random cut by combining an exponential mechanism with scoring based on the total AE after splitting.

Let $\mathcal{B}^{(p)}_L$ and $\mathcal{B}^{(p)}_R$ be the blocks split from $\mathcal{B}$ by the cutting point $p$, and the quality function $Q$ of $p$ in $\mathcal{B}$ is defined as follows:
\begin{equation}
Q(\mathcal{B},p) = - (\text{AE}(\mathcal{B}^{(p)}_L) + \text{AE}(\mathcal{B}^{(p)}_R) ).
\end{equation}
Then, we compute the score for all attribute values $p \in dom(a)$, $a \in A$ , and satisfies $|\mathcal{B}^{(p)}_L|\ge 1$ and $|\mathcal{B}^{(p)}_R|\ge 1$.
%$p$ is included in domains of $\mathcal{B}^{(p)}_L$ and 
Note that the number of candidates for $p$ is proportional to the sum of the domains for each attribute $\sum_{i\in[d]} |dom(a_i)|$, not to the total domains $\prod_{i\in[d]} |dom(a_i)|$.
We employ weighted sampling via an exponential mechanism to choose one cutting point $p*$. % from all attribute values by considering their scores.
The sampling probability of $p$ is proportional to  
\begin{equation}
\small
\label{eq:random_cut}
    \Pr[p*=p] \sim \exp\left(\frac{\epsilon Q(B,p)}{2\Delta_Q}\right)
\end{equation}
where $\Delta_Q$ is the L1-sensitivity of the quality metric $Q$.
We denote the L1-sensitivity of AE as $\Delta_{AE}$, and we can easily find $\Delta_Q=2\Delta_{AE}$ because $Q$ is the sum of two AEs.
Thus, each time a cut point $p$ is published according to such weighted sampling, a privacy budget of $\epsilon$ is consumed.
We set $\epsilon$ as the budget allocated to random cut (i.e., $(1-\gamma)\epsilon_r$) divided by $\kappa$.
If the cutting depth exceeds $\kappa$, we switch to random sampling (Line \ref{line:random_sampling} in Algorithm \ref{algo:proposed}).
Hence, cutting will not stop regardless of the depth or budget.

% Random cut has no constraint on the attributes to be selected, and one attribute is selected from all the attribute values each time to perform the bisection, which provides a very flexible form of partitioning for the $d$-dimensional count tensor.
% The quality function $Q$ is set to the sum of AEs to heuristically reduce the total AE and to reduce the number of cuts.
% For example, if the minimum AE value is used, there is a possibility that small blocks with small AE sizes are separated, but this can be prevented by using the sum.

Compared to Privtree \cite{zhang2016privtree}, for a $d$-dimensional block, at each cut, \method generates just 2 blocks with this random cut while Privtree generates $2^d$ blocks with fixed cutting points.
Privtree's heuristics prioritizes finer partitioning, which sufficiently works in low-dimensional data because AEs become very small and the total PEs is not so large.
In high-dimensional data, however, it causes unnecessary block splitting resulting in too much PEs.
\method carefully splits blocks one by one, thus suppressing unnecessary block partitioning and reducing the number of blocks i.e., smaller PEs.
It also enables flexibly shaped multidimensional block partitioning.
% The error trends by these methods are compared in detail in the experimental section, using low- (i.e., 2D) and our targeted high-dimensional dataset.
Moreover, while whole design of \method including convergence decision logic and cutting strategy are based on an error optimization problem as described in Section \ref{sec:problem_formulation}, Privtree has no such background.
This allows \method to provide effective block partitioning rather than simply fewer blocks, which we empirically confirm in Section \ref{sec:range_query}.

\subsection{Privacy Accounting}
\label{sec:privacy_accouting}
For privacy consumption accounting, since \method recursively splits a block into two disjoint blocks, we only have to trace a path toward convergence.
In other words, because \method manipulates all the blocks separately, we can track the total privacy consumption by the parallel composition for each converged block.
The information published by the recursive bisection is the result of segmentation; however, note that since there is a constraint on the cutting method for the block, it must be divided into two parts; in the worst case, the published blocks may expose all the cutting points.
For a given converged block $\mathcal{B}$, we denote the series of cutting points by $S_{\mathcal{B}}=[p_1,... . p_k]$, and $\mathcal{B}_{p_i}$ as the block after being divided into two parts at cutting point $p_i$.
To show the DP guarantee, let $D$ and $D'$ be the neighboring databases, and let $\Pr[S_{\mathcal{B}}|D]$ be the probability that $S_{\mathcal{B}}$ is generated from $D$.
We need to show that for any $D$, $D'$, and $S_{\mathcal{B}}$ that
\begin{equation}
\small
    \left|\frac{\Pr[S_{\mathcal{B}}|D]}{\Pr[S_{\mathcal{B}}|D']}\right| \le e^{\epsilon}
\end{equation}
to show that the recursive bisection satisfies $\epsilon$-DP.

The block with the largest $\left|\frac{\Pr[S_{\mathcal{B}}|D]}{\Pr[S_{\mathcal{B}}|D']}\right|$ of the converged disjoint blocks is $\mathcal{B}^*$, which has the longest $S_{\mathcal{B}^*}$ and contains different data between $D$ and $D'$.
Random converge and random cut are represented as follows.
\begin{align*}
\small
    &\left| \cfrac{\Pr[S_{\mathcal{B}^*}|D]}{\Pr[S_{\mathcal{B}^*}|D']}\right| =
        \cfrac{\Pr[\text{BAE}(\mathcal{B}_{p_0})+Lap(\lambda)>\theta]}{\Pr[\text{BAE}(\mathcal{B}'_{p_0})+Lap(\lambda)>\theta]} \\
        &\,\,\,\,\,\, \cdot \cfrac{\Pr[p*=p_1|D]}{\Pr[p*=p_1|D']}
        \cdot \cfrac{\Pr[\text{BAE}(\mathcal{B}_{p_1})+Lap(\lambda)>\theta]}{\Pr[\text{BAE}(\mathcal{B}'_{p_1})+Lap(\lambda) > \theta]} \\
        % \cfrac{\Pr[p*=p_2|D]}{\Pr[\Pr[p*=p_2|D']} \cdot
        % \cfrac{\Pr[bae(B_{p_2})+Lap(\lambda)>\theta|D]}{\Pr[bae(B'_{p_2})+Lap(\lambda) > \theta|D']} \cdot \cdots \cdot \\
        & \,\,\,\,\,\, \cdot \cdots \cdot  \cfrac{\Pr[p*=p_k|D]}{\Pr[p*=p_k|D']} \cdot
        \cfrac{\Pr[\text{BAE}(\mathcal{B}_{p_k})+Lap(\lambda) \le \theta]}{\Pr[\text{BAE}(\mathcal{B}'_{p_k})+Lap(\lambda) \le \theta]} \numberthis
\end{align*}
where $\mathcal{B}_{p_0}$ is the initial count tensor and for all $i$, $\mathcal{B}'_{p_i}$ indicates a neighboring block for $\mathcal{B}_{p_i}$.
Taking the logarithm,
{\footnotesize
\begin{align*}
    &\;\;\;\;\;\;\mathrm{ln}\left(\cfrac{\Pr[S_{\mathcal{B}^*}|D]}{\Pr[S_{\mathcal{B}^*}|D']}\right) = \underbrace{
        \sum_{i=1}^{k}{\mathrm{ln}\left(\cfrac{\Pr[p*=p_i|D]}{\Pr[p*=p_i|D']}\right)}}_{(*1) \mbox{ : for random cut}}  \notag \\
        &+ \underbrace{\sum_{i=0}^{k}{\mathrm{ln}\left(\cfrac{\Pr[\text{BAE}(\mathcal{B}_{p_i})+Lap(\lambda)>\theta]}{\Pr[\text{BAE}(\mathcal{B}'_{p_i})+Lap(\lambda)>\theta]}\right)} +  \mathrm{ln}\left(\cfrac{\Pr[\text{BAE}(\mathcal{B}_{p_k})+Lap(\lambda) \le \theta|D]}{\Pr[\text{BAE}(\mathcal{B}'_{p_k})+Lap(\lambda) \le \theta|D']}\right)}_{(*2) \mbox{ : for random converge}} \numberthis \label{eq:log_likeli} .
\end{align*}
}
% If $|\Pr[S|D]/\Pr[S|D']|<\mathrm{e}^{\epsilon_r}$ for any $D$,$D'$,$S$, then Recursive Bisection satisfies $\epsilon_r$-DP.
and let the first item of the right-hand of Eq.(\ref{eq:log_likeli}) be $(*1)$, and the other items be $(*2)$.

$(*1)$ corresponds to the privacy of the random cut, with each probability following Eq.(\ref{eq:random_cut}).
Given $\epsilon=\epsilon_{cut}$, for any $k$, the following holds from sequential composition.
\begin{align}
\label{eq:random_cut_gurantree}
\begin{aligned}
    \left|\sum_{i=1}^{k}{\mathrm{ln}\left(\cfrac{\Pr[p*=p_i|D]}{\Pr[p*=p_i|D']}\right)}\right| \le \kappa \epsilon_{cut}=(1-\gamma) \epsilon_{r} .
\end{aligned}
\end{align}

The following are privacy guarantees for the other part, $(*2)$, based on the observations presented in \cite{zhang2016privtree}.
First, we consider the sensitivity of AE $\Delta_{AE}$.

\begin{theorem}
The L1-sensitivity of the AE is $2(1-1/|\mathcal{B}|)$.
\end{theorem}

\begin{proof}
Let $\mathcal{B}'$ be the block that differs by only one count from $\mathcal{B}$.
The $\text{AE}(\mathcal{B}')$ can be computed as follows:
\begin{equation*}
\small
    \begin{split}
        \text{AE}(\mathcal{B}')
            &= \sum_{i \neq j \in [|\mathcal{B}|]} \left| x_i - \frac{S+1}{|\mathcal{B}|} \right| 
            + \left| x_j + 1 - \frac{S+1}{|\mathcal{B}|} \right| . \\
    \end{split}
\end{equation*}
Finally, the L1-sensitivity of $\text{AE}$ can be derived as:
\begin{equation*}
    \Delta_{AE} 
    = (|\mathcal{B}|-1) \frac{1}{|\mathcal{B}|} + 1 - \frac{1}{|\mathcal{B}|} \\
    = 2(1-1/n)
\end{equation*}
% Next, let $B'$ is the block that differs only two counts from $B$.
% This corresponds to a neighboring database where one data is turned into another data.
% The maximum change is when a count greater than the mean increases by one and a count less than the mean decreases by one.
% In this case, the maximum change is 2.
% Otherwise, if both counts are above the mean, the AE will not increase by more than 2 because the mean increases.
% The same is true when both are below the mean.
\end{proof}

Thus, we also obtain $|\text{BAE}(\mathcal{B}) - \text{BAE}(\mathcal{B}')| \le 2$, and 
\begin{align*}
\small
    (*2) &\le
        \sum_{i=0}^{k-1}{\mathrm{ln}\left(\cfrac{\Pr[\text{BAE}(\mathcal{B}_{p_i})+Lap(\lambda)>\theta]}{\Pr[\text{BAE}(\mathcal{B}_{p_i})-2+Lap(\lambda)>\theta]}\right)} \\
        & \quad + \mathrm{ln}\left(\cfrac{\Pr[\text{BAE}(\mathcal{B}_{p_k})+Lap(\lambda) \le \theta]}{\Pr[\text{BAE}(\mathcal{B}_{p_k})+2+Lap(\lambda) \le \theta]}\right) \numberthis \label{eq:upper_bound_by_sensitivity}.
\end{align*}
Furthermore, from the proof in the Appendix in \cite{zhang2016privtree}, when we have $f(x) = \mathrm{ln}\left(\frac{\Pr[x+Lap(\lambda)>\theta]}{\Pr[x-2+Lap(\lambda)>\theta]}\right)$, then
\begin{equation}
\label{eq:upper_bound_by_inequality}
\small
\begin{split}
\begin{cases}
    f(x) \le \frac{2}{\lambda}, & (\theta-x+2 > 0) \\
    f(x) \le \frac{2}{\lambda}\exp\left(\frac{\theta-x+2}{\lambda}\right), & (\theta-x+2 \le 0)
\end{cases}
\end{split}
\end{equation}

Next, we show the monotonic decreasing property of AE for block partitioning.
\begin{theorem}
\label{theorem:monotonic_ae}
For any $i=0,...,k-1$, $\text{AE}(\mathcal{B}_{p_i}) \ge \text{AE}(\mathcal{B}_{p_{i+1}})$.
\end{theorem}
\begin{proof}
We show that when $\mathcal{B}^{+}$ is an arbitrary block $\mathcal{B}$ with an arbitrary element $x$ $(>0)$ added to it, the AEs always satisfy $\text{AE}(\mathcal{B})\le AE(\mathcal{B}^{+})$.
Let the elements in $\mathcal{B}$ be $x_1, ..., x_k$, and let $\mathcal{B}^{+}$ be the block with $x_{k+1}$ added.
The mean values in each block are $\bar{x}=\frac{1}{k}(x_1+\cdots+x_k)$ and $\bar{x}^{+}=\frac{1}{k+1}(x_1+\cdots+x_{k+1})$ and $\text{AE}(\mathcal{B})=\sum_{i=1}^k{|x_i-\bar{x}|}$ and $\text{AE}(\mathcal{B}^{+})=\sum_{i=1}^{k+1}{|x_i-\bar{x}^{+}|}$.
Considering how much the AE can be reduced with the addition of $x_{k+1}$ to $\mathcal{B}$, $|x_i-\bar{x}|-|x_i-\bar{x}^{+}| \le |\bar{x}-\bar{x}^{+}|$ holds for each $i$ ($=1,...,k$), so $\text{AE}(\mathcal{B}) - \text{AE}(\mathcal{B}^{+})$ is at most $k \cdot |\bar{x}-\bar{x}^{+}|$.
On the other hand, with the addition of $x_{k+1}$, AE increases by at least $|x_{k+1}-\bar{x}^{+}|$ because this is a new item.
Since $x_{k+1}=(k+1)\bar{x}^{+} - (x_1+\cdots+x_k)=(k+1)\bar{x}^{+} - k\bar{x}$, then $|x_{k+1}-\bar{x}^{+}|=|k\cdot(\bar{x}-\bar{x}^{+})| = k\cdot|\bar{x}-\bar{x}^{+}|$.
Hence, $\text{AE}(\mathcal{B}^{+}) - \text{AE}(\mathcal{B}) \ge k \cdot |\bar{x}-\bar{x}^{+}| - k \cdot |\bar{x}-\bar{x}^{+}| = 0$ always holds.
Therefore, since $\mathcal{B}_{p_i}$ always has more elements than $\mathcal{B}_{p_{i+1}}$, $\text{AE}(\mathcal{B}_{p_i}) \ge \text{AE}(\mathcal{B}_{p_{i+1}})$.
\end{proof}

Considering $\text{BAE}(\mathcal{B})$, there exists a natural number $m$ ($1 \le m \le k$) where if $i<m$, $\text{BAE}(\mathcal{B}_{p_i})) \ge \text{BAE}(\mathcal{B}_{p_{i+1}})+\delta \ge \theta + 2 - \delta$, if $i=m$, $\theta + 2 \ge \text{BAE}(\mathcal{B}_{p_{i}})\ge \theta + 2 - \delta$, and if $m<i$,  $\text{BAE}(\mathcal{B}_{p_i}) = \theta + 2 - \delta$.
Therefore, using Eqs.(\ref{eq:upper_bound_by_sensitivity}, \ref{eq:upper_bound_by_inequality}),
\begin{equation}
\label{eq:bound_last}
\small
\begin{split}
    (*2) & \le
        \frac{2}{\lambda} + \sum_{i=1}^{m-1} \frac{2}{\lambda}\exp\left(\frac{\theta - \text{BAE}(\mathcal{B}_{p_i}) + 2}{\lambda}\right) + \frac{2}{\lambda} \\
        &\le
        \frac{4}{\lambda} + \frac{2}{\lambda} \cdot \frac{1}{1-\exp\left(-\frac{\delta}{\lambda}\right)} \\
        & = \frac{2}{\lambda} \cdot \frac{3\exp\left(\frac{\delta}{\lambda}\right)-2}{\exp\left(\frac{\delta}{\lambda}\right)-1}.
\end{split}
\end{equation}
Thus, to make $(*2)$ satisfy $\gamma \epsilon_r$-DP, $\frac{2}{\lambda} \cdot \frac{3\exp\left(\frac{\delta}{\lambda}\right)-2}{\exp\left(\frac{\delta}{\lambda}\right)-1} \le \gamma \epsilon_r$ should holds.
Since the $\lambda$ and $\delta$ that satisfy these conditions are not uniquely determined, these values are determined by giving $\exp(\delta / \lambda)$ as a hyperparameter $\alpha$.
Then, we can always calculate $\lambda = (\frac{3\alpha-2}{\alpha-1}) \cdot (\frac{2}{\gamma \epsilon_r})$ and $\delta = \lambda \log{\alpha}$, in turn, which satisfies $(*2) \le \gamma \epsilon_r$.
$\alpha$ is valid for $\alpha > 1$.
If $\alpha$ is extremely close to 1, $\lambda$ diverges and \textit{random convergence} is too inaccurate.
As $\alpha$ increases, $\lambda$ decreases, but $\delta$ increases.
Thus $\lambda$ and $\delta$ are trade-offs, and independently of the dataset, there exists a point at which both values are reasonably small.
Around $\alpha=1.4\sim1.8$ works well empirically.
Please see Appendix \ref{appendix:hyperparameters} for a specific analytical result.

Finally, together with $(*1)$, the recursive bisection by random converge and random cut satisfies $\epsilon_r$-DP.
In addition, the perturbation consumes $\epsilon_p$ for each block to add Laplace noise, so together with this, \method satisfies $\epsilon_p + \epsilon_r = \epsilon_b$-DP.

\subsection{Error Analysis}
When a p-view created by \method publishes a counting query answer, we can dynamically estimate an upper bound distribution of the error included in the noisy answer.
The upper bound of the error can be computed from the number of blocks used to answer the query and the distribution of the perturbation.
Note that this can be computed without consuming any extra privacy budget because, as shown in \ref{sec:privacy_accouting}, in addition to the count values, block partitioning results are released in a DP manner.

As a count query on the p-view is processed as Eq. (\ref{eq:count_query}), the answer consists of the sum of the query results for each block, and from Eq. (\ref{eq:two_errors}), each block contains two types of errors: AE and PE.
Let the error of a counting query $q$ be $Error(q, \mathcal{X}, \mathcal{X'}) := ||q(\mathcal{X}) -  q(\mathcal{X'})||_1 $ where $\mathcal{X}$ and $\mathcal{X'}$ are the original and noisy data, respectively, and we define the error by the L1-norm.
% Note that $Error(q, \mathcal{X}, \mathcal{X'})$ is a random variable and that the probability space is over the coin flips of the Laplace mechanism and random converge for all blocks.
% \begin{equation}
%     \label{eq:query_error_def}
%     % \begin{align}
%     Error(q, \mathcal{X}, \mathcal{X'}) := ||q(\mathcal{X}) -  q(\mathcal{X'})||_1 \\
%         % &=  \sum_{i=1,...,m}{AE_{q}(\mathcal{B}_i) + PE_{q}(\mathcal{B}_i)}
%     % \end{align}
% \end{equation}
First, since the AE depends on the concrete count values of each block involved in each query condition, we characterise the block distribution by defining an $\xi$-uniformly scattered block.
\begin{definition}[$\xi$-uniformly scattered]
A block $\mathcal{B}$ is $\xi$-uniformly scattered if for any subblock $\mathcal{B'} \subset \mathcal{B}$,
\begin{equation}
\small
AE(\mathcal{B'})/|\mathcal{B'}| \le \xi \cdot AE(\mathcal{B})/|\mathcal{B}| .
\end{equation}
\end{definition}
\noindent
While $\xi$ depends on the actual data, it is expected to decrease with each step by random cut.

Then, we have the following theorem for the error.
%\frac{1}{2} \text{e}^{\frac{\delta-2}{\lambda}}
\begin{theorem}
If for all $i$, block $\mathcal{B}_i$ is $\xi_i$-uniformly scattered, any $\mu$ satisfying $0 < \mu < 1$, and any $t$ satisfying $|t| < \epsilon_p$ and $|t| < \frac{1}{\lambda}$, the error of a counting query satisfies $Error(q, \mathcal{X}, \mathcal{X'}) \ge \Theta_{min}(\mu)$ and $Error(q, \mathcal{X}, \mathcal{X'}) \le \Theta_{max}(\mu)$ with probability of at least $1-\mu$, respectively, with
\begin{align}
\begin{aligned}
\small
\Theta_{min}(\mu) &= \frac{1}{t} \left( \log{\mu} + \sum_{i=1,...,m}{\log{(1-(\frac{w_i}{\epsilon_p})^2 t^2)}} \right) \\
\Theta_{max}(\mu) &= \sum_{i=1,...,m}{\xi_i w_i(k_i \delta + \theta)} \\
    &\; - \frac{1}{t} \left( \log{\mu} + \sum_{i=1,...,m}{\log{(1-(\frac{w_i}{\epsilon_p})^2 t^2)} + \log{(1-(\xi_i w_i\lambda)^2 t^2)}} \right)
\end{aligned}
\end{align}
where $w_i = \frac{ \left| \mathcal{B}_i \cap c_q \right| }{\left| \mathcal{B}_i \right|}$, $k_i$ is depth of $\mathcal{B}_i$ that can be public information.
\end{theorem}

\begin{proof}
The errors included in $Error(q, \mathcal{X}, \mathcal{X'})$ are PEs and AEs.
Both of them follow independent probability distributions for each block, and we first show the PE.
For each $\mathcal{B}_i$, perturbation noise is uniformly divided inside $\mathcal{B}_i$. 
Hence, the total PE in the query $q$ is represented by
$\sum_{i=1,...,m}{w_i* PE(\mathcal{B}_i) }$ where $w_i=\frac{ \left| \mathcal{B}_i \cap c_q \right| }{\left| \mathcal{B}_i \right|}$ and $PE(\mathcal{B}_i)$ is Laplace random variable following $Lap(\frac{1}{\epsilon_p})$.

Then, we consider the AE.
From random converge, given a $\mathcal{B}_i$, then $BAE(\mathcal{B}_i)+Lap(\lambda) \le \theta$ holds.
Considering $\text{BAE}(\mathcal{B}) = \mbox{max}(\theta + 2 - \delta,\,  \text{AE}(\mathcal{B})-k\delta )$,
when $\theta + 2 - \delta \le \text{AE}(\mathcal{B}_i) - k_{i}\delta$,
\begin{equation}
    \small
    \text{AE}(\mathcal{B}_i) \le Lap(\lambda) + k_{i}\delta + \theta .
\end{equation}
% Moreover,
% \begin{align}
% \small
%     &\,\,\, \Pr[\text{AE}(\mathcal{B}_i) \ge \lambda \log{\frac{1}{2\mu}} + k_{i}\delta + \theta] \\
%     &\le \Pr[Lap(\lambda) + k_{i}\delta + \theta \ge \lambda \log{\frac{1}{2\mu}} + k_{i}\delta + \theta] \\
%     &= \frac{1}{2} \Pr[|Lap(\lambda)| \ge \lambda \log{\frac{1}{2\mu}} ] \\
%     &= \frac{1}{2} \text{e}^{\log{2\mu}} = \mu .
% \end{align}
And when $\theta + 2 - \delta > \text{AE}(\mathcal{B}_i) - k_{i}\delta$, 
% \begin{align}
%     &(\lambda \log{\frac{1}{2\mu}} + k_{i}\delta + \theta) - (\theta + 2 - \delta + k_{i}\delta) \\
%     &= \delta - 2 + \lambda \log{\frac{1}{2\mu}} \ge 0 \;\;\; (\because \mu < \frac{1}{2} \text{e}^{\frac{\delta-2}{\lambda}}) . 
% \end{align}
\begin{equation}
    \small
    \text{AE}(\mathcal{B}_i) - k_i \delta < \theta + 2 - \delta \le Lap(\lambda) + \theta
\end{equation}
Thus, the upper bound of $\text{AE}(\mathcal{B}_i)$ is distributed under $Lap(\lambda) + k_{i}\delta + \theta$.
In other words, AE cannot be observed directly, but the upper bound distribution is bounded by the Laplace distribution.
Also note that the AE satisfies $\text{AE}(\mathcal{B}_i) \ge 0 $.

Therefore, for the error lower bound, we only need to consider the $m$ PEs, $\sum_{i=1,...,m}{w_i* PE(\mathcal{B}_i) }$.
$PE(\mathcal{B}_i)$ is independent random variable, respectively.
We apply Chernoff bound to the sum, for any $a$ and $t$,
\begin{equation}
    \label{eq:chernoff_1}
    \small
    \Pr\left[Error(q, \mathcal{X}, \mathcal{X'}) \le a\right] \le e^{(ta)}\prod_{i=1,...,m}{E[e^{(-t w_i PE(\mathcal{B}_i))}]},
\end{equation}
where $|t|<\epsilon_p$ is required for existence of the moment generating function.
By using $PE(\mathcal{B}_i)$ follows $Lap(\frac{1}{\epsilon_p})$, we can derive
\begin{equation}
    \small  
    \Pr\left[Error(q, \mathcal{X}, \mathcal{X'}) \le \frac{1}{t} \left( \log{\mu} + \sum_{i=1,...,m}{\log{(1-(\frac{w_i}{\epsilon_p})^2 t^2)}} \right)\right] \le \mu .
\end{equation}
% 解析的に最適なtを計算できなかった，やれと言われたら対応する
% Let $f(t) = \frac{1}{t} \left( \log{\mu} + \sum_{i=1,...,m}{\log{(1-(\frac{w_i}{\epsilon_p})^2 t^2)}} \right)$, 

On the other hand, for the error upper bound, we need to consider AEs as well.
Hence, we apply Chernoff bound to the sum of $2m$ independent random variables following each Laplace distribution.
Considering the upper bound distribution of $AE(\mathcal{B}_i)$ has $w_i(k_i\delta + \theta)$ for the mean and $\lambda$ for the variance, let $\bar{AE}(\mathcal{B}_i)$ be a Laplace random variable whose mean and variance are 0 and $\lambda$, respectively, then we have
\begin{align}
\begin{aligned}
    \small  
    \Pr\Bigl[&Error(q, \mathcal{X}, \mathcal{X'}) - \sum_{i=1,...,m}{\xi_i w_i(k_i\delta + \theta)} \ge a \Bigr] \\
    &\le e^{(-ta)}\prod_{i=1,...,m}{E[e^{(t w_i PE(\mathcal{B}_i))}]E[e^{(t w_i \xi_i \bar{AE}(\mathcal{B}_i))}]}, 
\end{aligned}
\end{align}
where since block $\mathcal{B}_i$ is $\xi_i$-uniformly scattered, AE included in the query $q$ and block $\mathcal{B}_i$ is at most $\xi_i w_i AE(\mathcal{B}_i)$.
Lastly, for any $t$, where $|t| < \epsilon_p$ and $|t| < \frac{1}{\lambda}$, from the inequality, we can derive as follows:
\begin{align*}
    \small
    &\Pr \Biggl[Error(q, \mathcal{X}, \mathcal{X'}) \ge \sum_{i=1,...,m}{\xi_i w_i(k_i \delta + \theta)} \\
    &\;\;\; - \frac{1}{t} \left( \log{\mu} + \sum_{i=1,...,m}{\log{(1-(\frac{w_i}{\epsilon_p})^2 t^2)} + \log{(1-(\xi_i w_i\lambda)^2 t^2)}} \right)\Biggr] \\
    &\le \mu .
\end{align*}
This completes the proof.
\end{proof}
\noindent
Importantly, this can be dynamically computed for any counting queries, helping the analyst to perform a reliable exploration.

Similarly, since the HDMM \cite{mckenna2018optimizing} optimizes budget allocations for counting queries by the MM, we can statically calculate the error distributions for each query.
However, this is workload-dependent.
In data exploration, we consider predefined workload is strong assumption to be avoided.

\section{Evaluation}

In this section, we report the results of the experimental evaluation of \method.
To evaluate our proposed method, we design experiments to answer the following questions:
\begin{itemize}
    \item How effectively can the constructed p-views be used in data exploration via various range counting queries?
    \item How space-efficiently can the constructed p-views represent high-dimensional count tensors?
    % \item How scalable is building the p-view with increasing dimensionality?
    % \item How effective can the constructed p-view generate synthetic samples for data mining tasks?
\end{itemize}
% As a whole, 
We shows the effectiveness of \method via range counting queries in section \ref{sec:range_query}, % and \textit{sampling from p-view} (section \ref{sec:classification}) via classification task.
and section \ref{sec:space_efficiency} reports the space efficiency.

\begin{table}[t]
    \centering
    \caption{Datasets.}
    \small
    \begin{tabular}{lrrrrr}
    \toprule
    \multicolumn{1}{l}{Dataset} & \multicolumn{1}{l}{\#Record} & \multicolumn{1}{l}{\begin{tabular}[c]{@{}c@{}}\#Column \\ (categorical) \end{tabular}} & \multicolumn{1}{l}{\#Domain} & Variance \\
    \midrule
    {\tt Adult} \cite{dataset_adult}         & 48842    & 15 (9)        & $9\times10^{19}$ & 0.0360\\
    {\tt Small-adult}	 & 48842    & 4 (2)         & $3\times10^{5}$ & 0.0237\\
    {\tt Numerical-adult}& 48842    & 7 (1)         & $2\times10^{11}$ & 0.0200 \\
    {\tt Traffic} \cite{dataset_traffic}    & 48204	& 8 (2)         & $1 \times 10^{14}$ & 0.0484\\
    {\tt Bitcoin} \cite{dataset_bitcoin}    & 500000   & 9 (1)         & $4 \times 10^{12}$ & 0.0379\\
    {\tt Electricity} \cite{dataset_electricity}	 & 45312	& 8 (1)         & $1 \times10^{14}$ & 0.0407\\
    {\tt Phoneme} \cite{dataset_phoneme}	     & 5404	    & 6 (1)         & $2 \times10^{6}$ & 0.0304\\
    {\tt Jm1} \cite{dataset_jm1}       & 10885    & 22 (1)        & $2 \times 10^{21}$ & 0.0027\\
    \bottomrule
    \end{tabular}
\label{tbl:dataset}
\end{table}

\subsection{Experimental Setup} \label{sec:exp_setup}

%We here describe the experimental setups.
We describe the experimental setups.
In the following experiments, we run 10 trials with \method and the competitors and report their averages to remove bias.
Throughout the experiments, the hyperparameters of \method are fixed as $(\epsilon_{r}/\epsilon_{b}, \alpha, \beta, \gamma) = (0.9, 1.6, 1.2, 0.9)$. 
Please see Appendix \ref{appendix:hyperparameters} for insights on how to determine the hyperparameters.
We provide observations and insights into all the hyperparameters of \method in Appendix of the full version \cite{fullversion}.
% These parameters work well independent of the datasets.
% Intuitively, even if we use most of the budget for the partitioning part, blocks that converge at shallow depths also can use large budgets by adaptive perturbation, which works effectively.

\noindent
\textbf{Datasets.}
We use several multidimensional datasets commonly used in the literature, as shown in Table \ref{tbl:dataset}.
{\tt Adult} \cite{dataset_adult} includes 6 numerical and 9 categorical attributes.
We prepare {\tt Small-adult} by extracting 4 attributes (age, workclass, race, and capital-gain) from {\tt Adult}.
Additionally, we form {\tt Numerical-adult} by extracting only numerical attributes and a label.
{\tt Traffic} \cite{dataset_traffic} is a traffic volume dataset.
{\tt Bitcoin} \cite{dataset_bitcoin} is a Bitcoin transaction graph dataset.
{\tt Electricity} \cite{dataset_electricity} is a dataset on changes in electricity prices.
{\tt Phoneme} \cite{dataset_phoneme} is a dataset for distinguishing between nasal and oral sounds.
{\tt Jm1} \cite{dataset_jm1} is a dataset of static source code analysis data for detecting defects with 22 attributes.
\method and most competitors require the binning of all numerical attribute values for each dataset.
Basically, we set the number of bins to 100 or 10 when the attribute is a real number.
We consider that the number of bins should be determined by the level of granularity that analysts want to explore, regardless of the distribution of the data.
For categorical columns, we simply apply ordinal encoding.
In Table \ref{tbl:dataset}, \#Domain shows the total domain sizes after binning.
Variance is the mean of the variance for each dimension of the binned and normalized dataset and gives an indication of how scattered the data is.

\noindent
\textbf{Implementations of competitors.}
We compare our proposed method \method with Identity \cite{dwork2006differential}, Privtree \cite{zhang2016privtree}, HDMM \cite{mckenna2018optimizing}, PrivBayes \cite{zhang2017privbayes}, and DAWA partitioning mechanism \cite{li2014data}.
% However, not all competitors have enough capacity for multi-dimensional data.
For these methods, we perform the following pre- and postprocessing steps.
% For Identity, we directly perturb low-dimensional data and compute errors.
% On the dataset with high dimensionality, we estimate errors following \cite{mckenna2018optimizing}, employing implicit matrix representations and workload-based estimation.
For Identity, we estimate errors following \cite{mckenna2018optimizing}, employing implicit matrix representations and workload-based estimation, because it is infeasible to add noises on a high-dimensional count tensor  because of the huge space.
% We refer to the Identity with the estimation as Identity\_est.
For Privtree, as described in \cite{zhang2016privtree}, we set the threshold to 0 and allocate half of the privacy budget to tree construction and half to perturbation.
Using the same method as \method, the blocks obtained by Privtree are used as the p-view.
For the HDMM, we utilize p-Identity strategies as a template.
DAWA's partitioning mechanism can be applied to multidimensional data by flattening data into a 1D vector.
However, when the domain size becomes large, the optimization algorithm based on the v-optimal histogram for the count vector cannot be applied due to the computational complexity.
Therefore, we apply DAWA to {\tt Small-adult} and {\tt Phoneme} because their domain sizes are relatively small.
We perform only DAWA partitioning without workload optimization to compare the partitioning capability without a given workload to evaluate workload-independent p-view generation.
For fairness, PrivBayes is trained on raw data\footnote[8]{PrivBayes shows worse performances with binned data in our prestudy.}.
PrivBayes, in counting queries, samples the exact number of original data points; therefore, it may consume extra privacy budget.

% \begin{table}[t]
%     \centering
%     \small
%     \begin{tabular}{lrrrrr}
%     \toprule
%     \multicolumn{1}{l}{Dataset} & \multicolumn{1}{l}{\#Record} & \multicolumn{1}{l}{\begin{tabular}[c]{@{}c@{}}\#Column \\ (categorical) \end{tabular}} & \multicolumn{1}{l}{\#Domain} & \multicolumn{1}{l}{\%Positive} \\
%     \midrule
%     {\tt Adult}\footnotemark[2]	         & 48842    & 15 (9)        & $9\times10^{19}$      & 0.239 \\
%     {\tt Small-adult}	 & 48842    & 4 (2)         & $3\times10^{5}$       & N/A   \\
%     {\tt Numerical-adult}& 48842    & 7 (1)         & $2\times10^{11}$      & 0.239 \\
%     {\tt Traffic}\footnotemark[3]	     & 48204	& 8 (2)         & $1 \times 10^{14}$    & N/A   \\
%     {\tt Bitcoin}\footnotemark[4]	     & 500000   & 9 (1)         & $4 \times 10^{12}$    & 0.986 \\
%     {\tt Electricity}\footnotemark[5]	 & 45312	& 8 (1)         & $1 \times10^{15}$     & 0.425 \\
%     {\tt Phoneme}\footnotemark[6]	     & 5404	    & 6 (1)         & $2 \times10^{6}$      & 0.293 \\
%     {\tt Jm1}\footnotemark[7]            & 10885    & 22 (1)        & $2 \times 10^{21}$    & N/A   \\
%     \bottomrule
%     \end{tabular}
%     \caption{Datasets.}
% \label{tbl:dataset}
% \end{table}

\noindent
\textbf{Workloads.}
We prepare different types of workloads.
\textit{$k$-way All Marginal} is all marginal counting queries using all combinations of $k$ attributes.
\textit{$k$-way All Range} is the range version of the marginal query.
\textit{Prefix-$k$D} is a prefix query using all combinations of $k$ attributes.
\textit{Random-$k$D Range Query} is a range query for arbitrary $k$ attributes and we randomly generate 3000 queries for a single workload.
% In the following experiments, we randomly issue 3000 queries for each workload. 
% These workloads potentially include other types of queries.
% Considering the diversity of queries in data exploration, we place particular emphasis on this result.

\noindent
\textbf{Reproducibility.}
The experimental code is publicly available on the \url{https://github.com/FumiyukiKato/HDPView}.
% We will make the experimental code public.
% For double-blind review, the code is on anonymous repository\footnote{\url{https://anonymous.4open.science/r/DP-Mondrian-CDD3}}.

%\input{05a_figures}

%\begin{table}[t]
%    \centering
%    \small
%    \begin{tabular}{lr}
%    \toprule
%    Algorithm & ARR \\
%    \midrule
%    \method      & $\mathbf{1.00}$ \\
%    Privbayes    & $\mathbf{5.00}$ \\
%    HDMM	     & $\mathbf{30.76}$ \\
%    Identity\_est	 & $\mathbf{16724395.39}$ \\
%    DAWA         & $\mathbf{3.95}$ \\
%    Identity     & $\mathbf{1.39}$ \\
%    \bottomrule
%    \end{tabular}
%    \caption{\method provides low-error counting queries on average in different workloads and datasets.}
%\label{tbl:query_arr}
%\end{table}

% \subsection{Range Counting Queries} \label{sec:range_query}

\subsection{Effectiveness} \label{sec:range_query}
% \begin{figure*}[t]
% % \begin{minipage}{0.48\hsize}
%     \centering
%     \includegraphics[width=0.98\hsize]{img/img_query_rmse_high.png}
%     \caption{\method shows small errors for a wide variety of high-dimensional range counting queries.}
%     \label{fig:query_rmse_high}
% % \end{minipage}
% \end{figure*}

% \begin{figure*}[t]
% % \begin{minipage}{0.48\hsize}
%     \centering
%     \includegraphics[width=0.98\hsize]{img/img_query_rmse_low.png}
%     \caption{In queries with a small number of attributes, \method outperforms the baseline (Identity), but HDMM, which includes the workload optimization, performs better than our method.}
%     \label{fig:query_rmse_low}
% % \end{minipage}
% \end{figure*}

\begin{figure*}[t]
% \begin{minipage}{0.48\hsize}
    \centering
    \includegraphics[width=0.98\hsize]{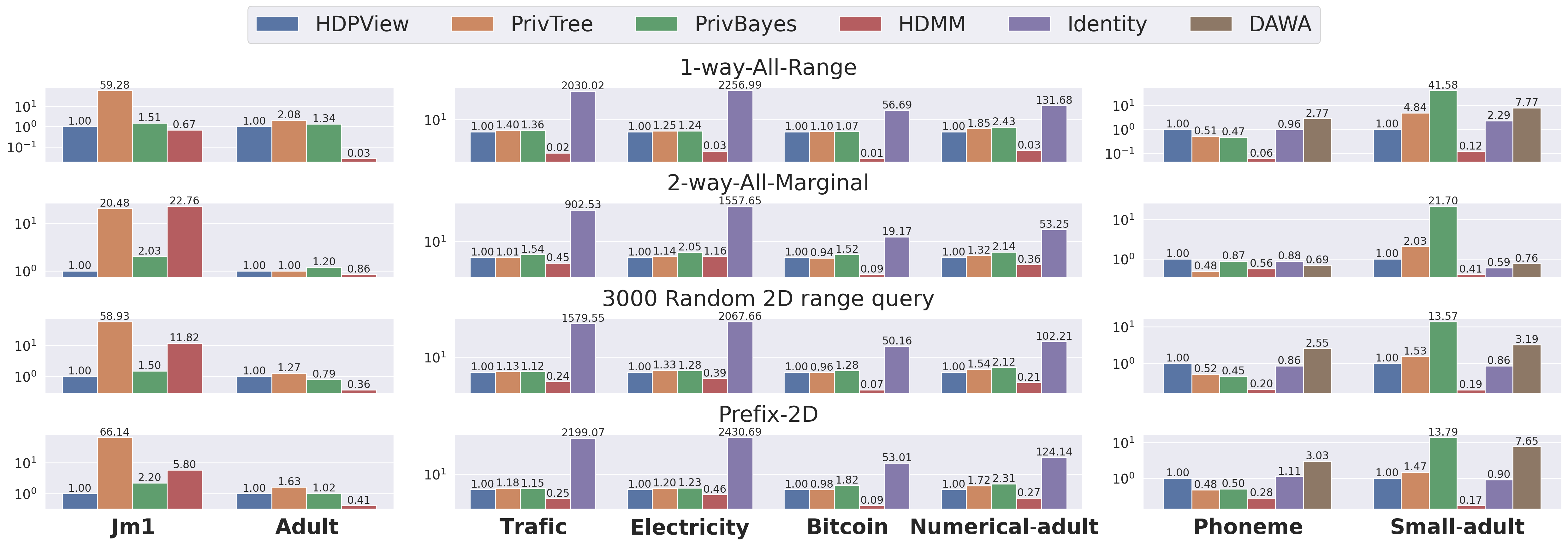}
    \includegraphics[width=0.98\hsize]{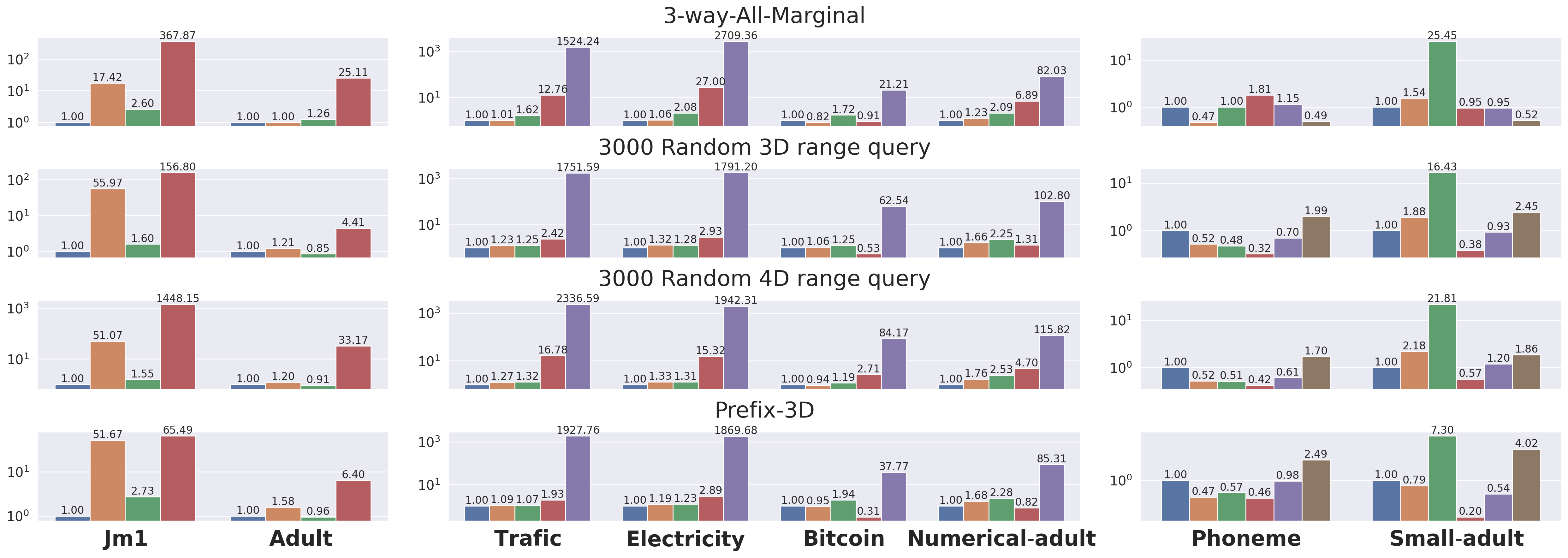}
    \caption{Relative RMSEs against \method on all the datasets and workloads: \method shows small errors for a wide variety of high-dimensional range counting queries.}
    \label{fig:query_rmse}
% \end{minipage}
\end{figure*}

% \begin{figure*}[t]
% % \begin{minipage}{0.48\hsize}
%     \centering
%     \includegraphics[width=0.98\hsize]{img/img_query_rmse_low.png}
%     \caption{In queries with a small number of attributes, \method outperforms the baseline (Identity), but HDMM, which includes the workload optimization, performs better than our method.}
%     \label{fig:query_rmse_low}
% % \end{minipage}
% \end{figure*}

We evaluate how effective p-views constructed by \method are in data exploration by issuing various range counting queries.
% We utilize the abovementioned workloads.

\noindent
\textbf{Evaluation metrics.}
We evaluate \method and other mechanisms by measuring the RMSE for all counting queries.
Formally, given the count tensor $\mathcal{X}$, randomized view $\mathcal{X'}$ and workload $\mathbf{W}$, the RMSE is defined as: $\mbox{RMSE} = \sqrt{\frac{1}{|\mathbf{W}|}\sum_{q \in \mathbf{W}}(q(\mathcal{X})-q(\mathcal{X'}))^{2}}$.
% \begin{equation}
% \small
% \nonumber
%   \mbox{RMSE} = \sqrt{\cfrac{1}{|\mathbf{W}|}\sum_{q \in \mathbf{W}}(q(\mathcal{X})-q(\mathcal{X'}))^{2}} .
% \end{equation}
This metric is useful for showing the utility of the p-view.
It corresponds to the objective function optimized by MM families \cite{li2015matrix, mckenna2018optimizing}, where given a workload matrix $\bm{W}$ and a query strategy $\bm{A}$, which is the optimized query set to answer the workload, the expected error of the workloads is $\frac{2}{\epsilon^2} ||\bm{A}||_{1}^2||\bm{W}\bm{A}^{+}||^2_{F} = \mbox{RMSE}^2$.
% \begin{equation}
% \small
% \nonumber
%   \cfrac{2}{\epsilon^2} ||\bm{A}||_{1}^2||\bm{W}\bm{A}^{+}||^2_{F} = \mbox{RMSE}^2 .
% \end{equation}
Thus, we can compare the measured errors with this optimized estimated errors.
We also report the relative RMSE against \method for comparison.
% To evaluate the average performance for a wide range of queries and datasets, we compute the average over all workloads and all datasets.
% We refer to this metric as \textit{average relative RMSE} (ARR).

\noindent
\textbf{High quality on average.}
Figure \ref{fig:query_rmse} shows the relative RMSEs for all datasets and workloads and algorithms with privacy budget $\epsilon$=1.0.
The relative RMSE (log-scale) is plotted on the vertical axis and the dataset on the horizontal axis where high-dimensional datasets ({\tt Jm1} and {\tt Adult}) are on the left, medium-dimensional datasets ({\tt Traffic}, {\tt Electricity}, {\tt Bitcoin} and {\tt Numerical-adult}) are in the middle, and low-dimensional datasets ({\tt Small-adult} and {\tt Phoneme}) are on the right.
The errors with Identity for high-dimensional data are too large and are omitted for appearance.
% DAWA's results are reported only on low-dimensional datasets since it lacks the capacity to address high-dimensional data.
As a whole, \method works well.
In Section \ref{sec:intro}, Table \ref{tbl:query_arr} shows the relative RMSE averaged over all workloads and all datasets in Figure \ref{fig:query_rmse}, and \method achieves the lowest error on average.
In data exploration, we want to run a variety of queries, so the average accuracy is important.
We believe \method has such a desirable property.
A detailed comparisons with the competitors are explained in the following paragraphs.

\noindent
\textbf{Comparison with Identity, HDMM and DAWA.}
Identity, which is the most basic Baseline, and HDMM, which performs workload optimization, cause more errors for high-dimensional datasets than \method.
For Identity, the reason is that the accumulation of noise increases as the number of domains increases.
\method avoids the noise accumulation by grouping domains into large blocks.
The results of HDMM show that the increasing dimension of the dataset and the dimension of the query can increase the error.
This is because the matrix representing the counting queries to which the matrix mechanism is applied becomes complicated, making it hard to find efficient budget allocations.
This is why the accuracy of the 3- or 4-dimensional queries for {\tt Jm1} and {\tt Adult} is poor with HDMM.
In particular, the HDMM's sensitivity to dimensionality increases can also be seen in Figure \ref{fig:query_add_attr}.
DAWA's partitioning leads more errors than the \method and Privtree.
When applied to multi-dimensional data, DAWA finds the optimal partitioning on a domain mapped in one-dimension, while \method and Privtree finds more effective multi-dimensional data partitioning.

% Therefore, compared to PrivBayes, \method's performance is limited in such a point.

% \method demonstrates stably small errors from low- to high-dimensional attributes, and PrivBayes does as well, but \method is better than PrivBayes under low dimensionality.
% PrivBayes might have difficulty learning enough correlations to approximate data distributions.
% On the other hand, the errors of HDMM simply increase with increasing data dimensionality.

\begin{figure}[t]
\begin{minipage}[t]{0.48\hsize}
    \centering
    \includegraphics[width=\hsize]{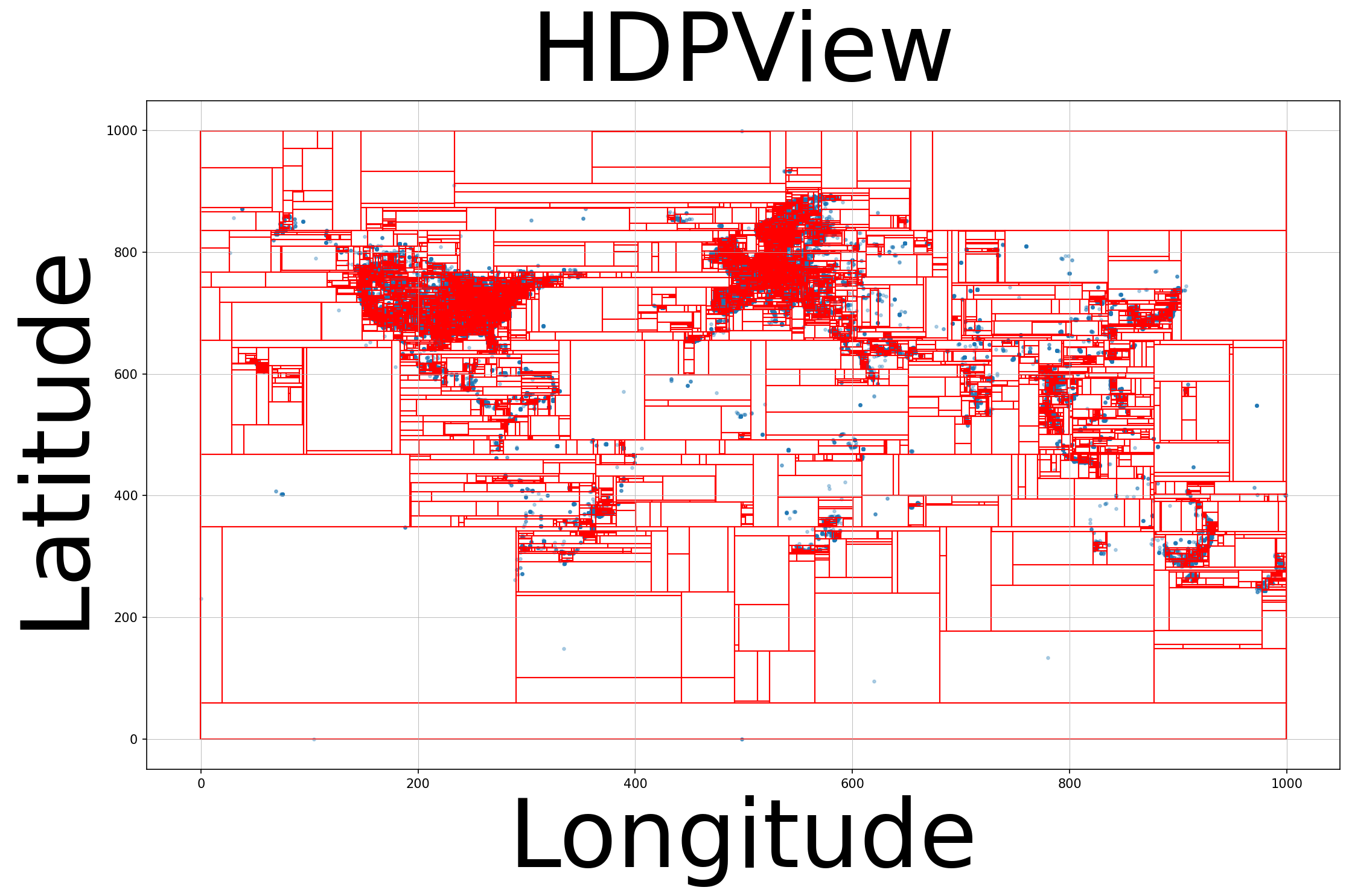}
    \vfill
    \begin{tabular}{ll}
        \hline
        \#Blocks & 15669  \\ \hline
        RMSE    & 657.37 \\ \hline
    \end{tabular}
\end{minipage}
\hfill
\begin{minipage}[t]{0.48\hsize}
    \centering
    \includegraphics[width=\hsize]{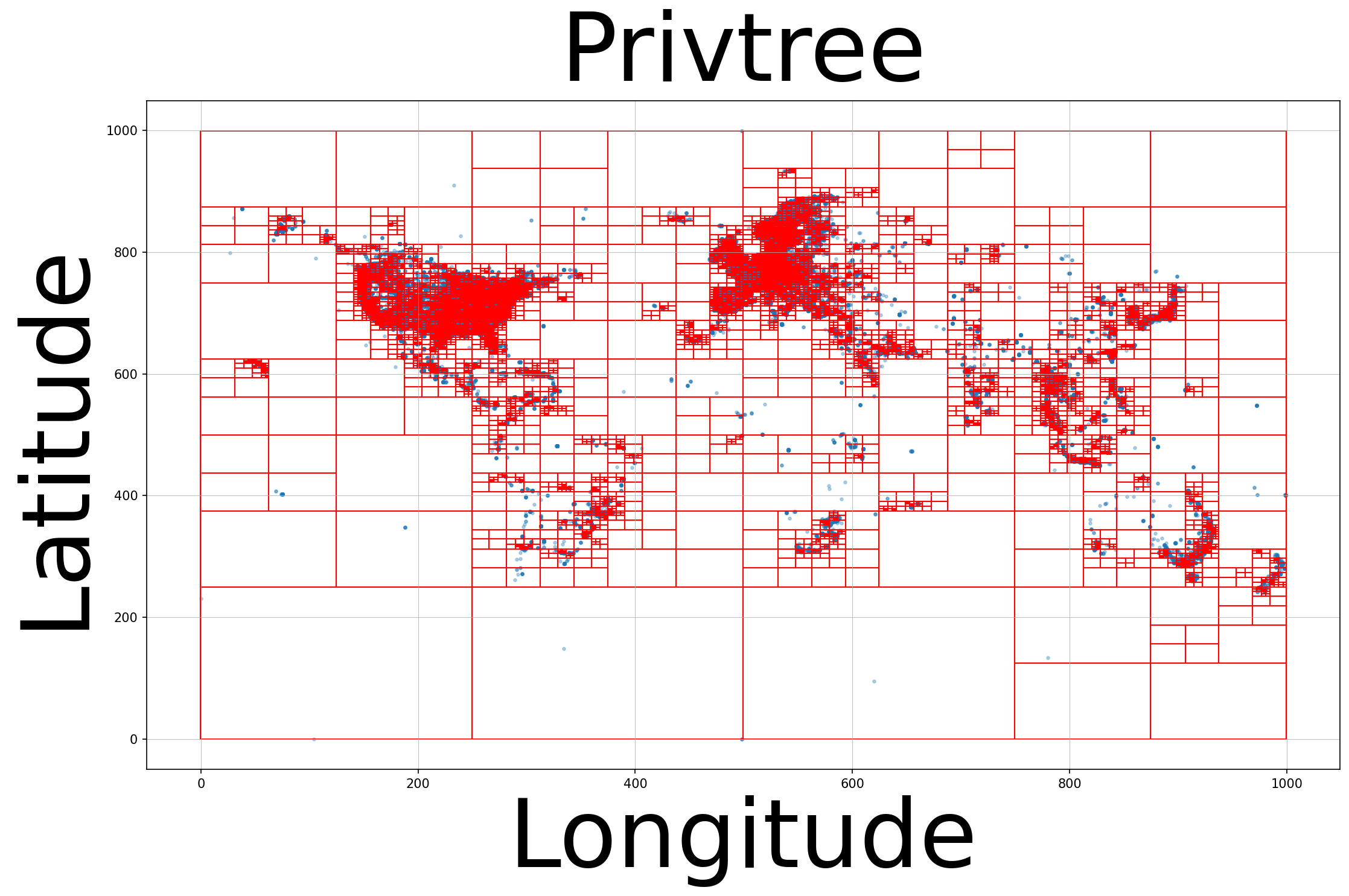}
    \vfill
    \begin{tabular}{ll}
        \hline
        \#Blocks & 19475  \\ \hline
        RMSE    & 315.79 \\ \hline
    \end{tabular}
\end{minipage}
\caption{Examples of \method (left) and Privtree (right) on 2D dataset ({\tt Gowalla}): \method has fewer blocks, leading to noisier results than Privtree for very low-dimensional data. Also, \method provides flexible block partitioning.}
\label{fig:2d-gowalla}
\end{figure}

\begin{figure}[t]
    \centering
    \includegraphics[width=\hsize]{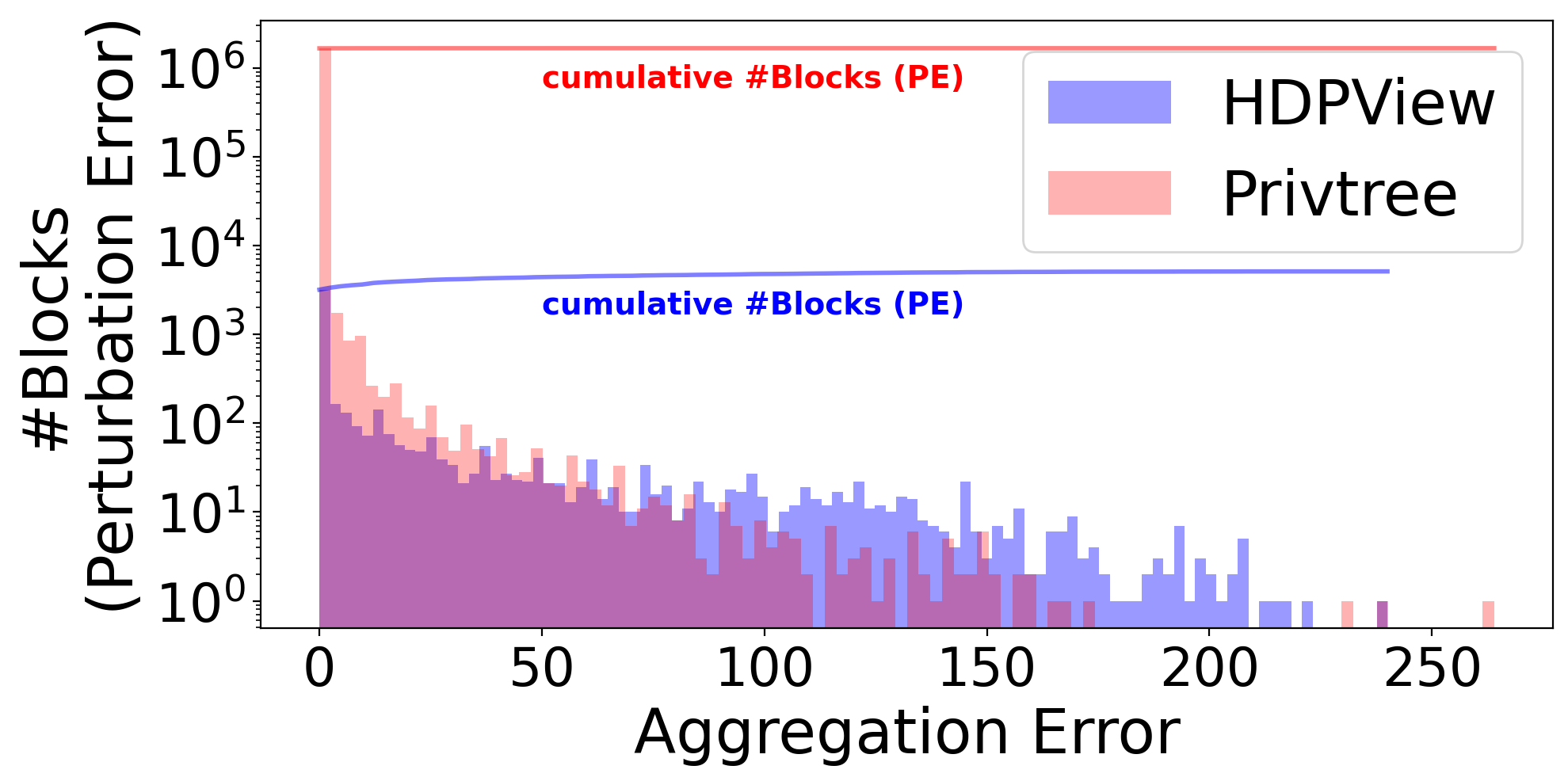}
    \caption{Number of blocks (log-scale) with various AEs for high-dimensional dataset (i.e., {\tt Adult}) for \method and Privtree. \method has slightly larger AE blocks, but Privtree has a much more number of blocks i.e., much larger PEs.}
\label{fig:adult_blocks}
\end{figure}

\begin{figure}[t]
    \centering
    \includegraphics[width=\hsize]{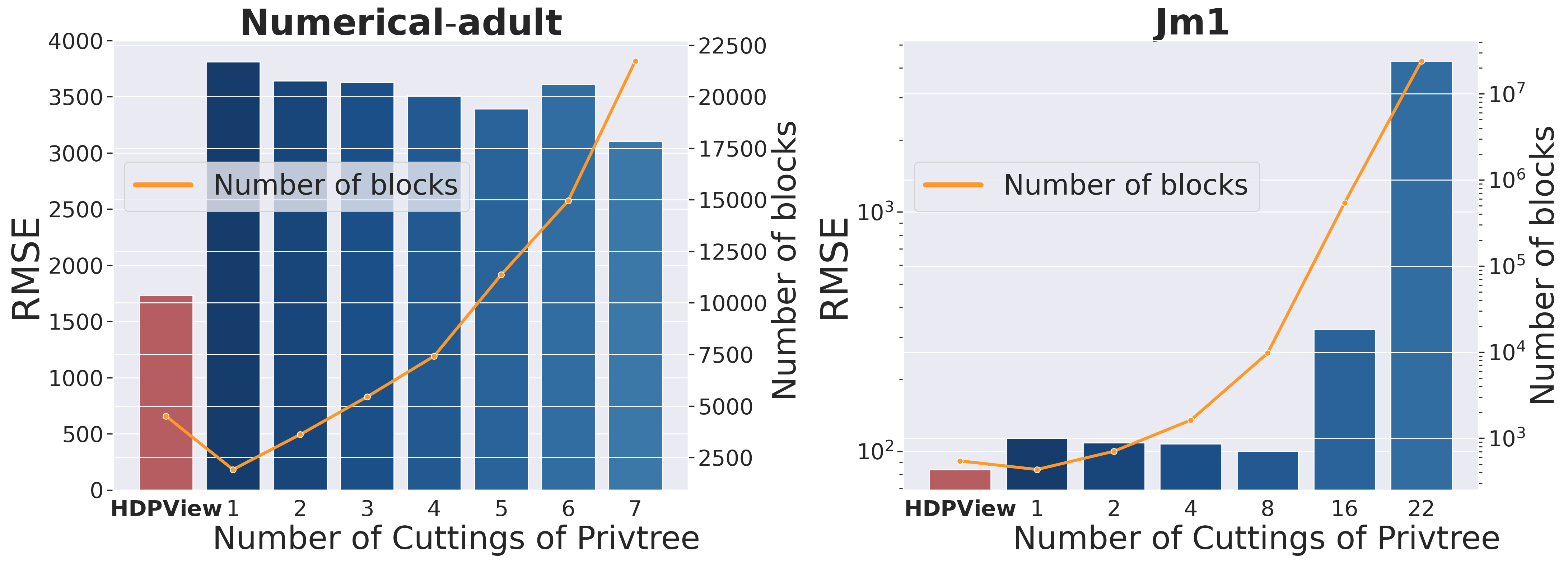}
    \caption{\method is more effective than Privtree even with controlled number of cuttings on {\tt Numerical-adult} (left) and {\tt Jm1} (log-scale) (right).}
    \label{fig:round_robin}
\end{figure}

\noindent
\textbf{Comparison with Privtree.}
Overall, \method outperforms Privtree's accuracy mainly for mid- to high-dimensional datasets.
In particular, we can see Privtree's performance drops drastically in high-dimensionality (i.e., {\tt Jm1}).
Privtree achieves higher accuracy than \method for {\tt Phoneme}.
This is likely because Privtree's strategy, which prioritize finer splitting, are sufficient for the small domain size rather than \method's heuristic algorithm.
Even if the blocks is too fine, the accumulation of PEs is not so large in low-dimensionality, and AEs become smaller, which results in an accurate p-view.
% The unstable performance in {\tt Small-adult} in Privtree may be related to the cardinality of the attributes ({\tt Small-adult}: (74, 9, 5, 100) and {\tt Phoneme}: (10, 10, 10, 10, 10, 10, 2)).
The reason why \method is better for {\tt Small-adult} despite the low-dimensionality may be that the sizes of the cardinality of attributes are uneven ({\tt Small-adult}: \{74, 9, 5, 100\}, {\tt Phoneme}: \{10, 10, 10, 10, 10, 10, 2\}), which may make Privtree's fixed cutting strategy ineffective.
% If the cardinality is not uniform, Privtree's cutting is affected, whereas the \method performs consistently because it chooses the cutting points among all attribute values.
To see the very low-dimensional case, Figure \ref{fig:2d-gowalla} shows the block partitioning for the 2D data with a popular {\tt Gowalla} \cite{dataset_gowalla} check-in dataset.
The table below shows the number of blocks and the RMSE for the 3000 Random 2D range query.
\method yields fewer blocks and Privtree generates a less noisy p-view for the abovementioned reason.
The figure also confirms that \method performs a flexible shape of block partitioning.

On the other hand, for high-dimensional dataset, this property can be avenged.
In Privtree, a single cutting always generates $2^d$ new blocks, which are too fine, resulting in very large PEs even though the AEs are smaller.
Figure \ref{fig:adult_blocks} shows the distribution of AEs for blocks on {\tt Adult} for \method and Privtree.
\method has slightly larger AE blocks, but Privtree has a large number of blocks and cause larger PEs.
An extreme case is {\tt Jm1} in which Privtree causes large errors. 
This is probably because {\tt Jm1} actually requires fewer blocks since the distribution is highly concentrated (c.f., Table \ref{tbl:dataset}).
Figure \ref{fig:block_size_comparison} shows that the number of blocks of generated p-view by \method and Privtree.
For {\tt Jm1}, \method generates very small number of blocks while Privtree does not.
We can confirm that \method avoids unnecessary splitting via \textit{random cut} and suppresses the increase in PEs which causes in Privtree.
This would be noticeable for datasets with concentrated distributions, where the required number of blocks is essentially small.
% In summary, Privtree's partitioning strategy contributes to smaller AEs, but leads to an excessive increase in PEs for high-dimensional data.

Figure \ref{fig:round_robin} shows the results of reducing the number of cut attributes in Privtree and adjusting the number of blocks in p-view on {\tt Numerical-adult} and {\tt Jm1}.
If the number of cut attributes is smaller than the dimension $d$, we choose target attributes in a round-robin way (Appendix of \cite{zhang2016privtree}).
In the case of {\tt Numerical-adult}, the error basically decreases as the number of cut attributes is increased, similar to the observation in Appendix of \cite{zhang2016privtree}.
However, for high-dimensional data such as {\tt Jm1}, the error increases rapidly as the number of cut attributes increases to some extent.
This is consistent with the earlier observation that influence of PEs increases.
Also, in any cases, the error of \method is smaller, indicating that \method not only has a smaller number of blocks, but also performs effective block partitioning compared to Privtree on these datasets.

% Furthermore, compared to {\tt Small-adult}, where the cardinalities of attributes are (74, 9, 5, 100), {\tt Phoneme} are (10, 10, 10, 10, 10, 10, 2).
% Because Privtree bisects all dimensions in a single cut, {\tt Phoneme} requires fewer cuts to obtain fine-grained blocks and results in fewer AE.

% This can be attributed to the higher density, i.e., more data in the small domain.
% Compared to \method, Privtree performs deeper block partitioning on denser data.
% On the other hand, for data with a large domain size, the accuracy is degraded, e.g., {\tt Jm1}.

\begin{figure}[t]
    \centering
    \includegraphics[width=\hsize]{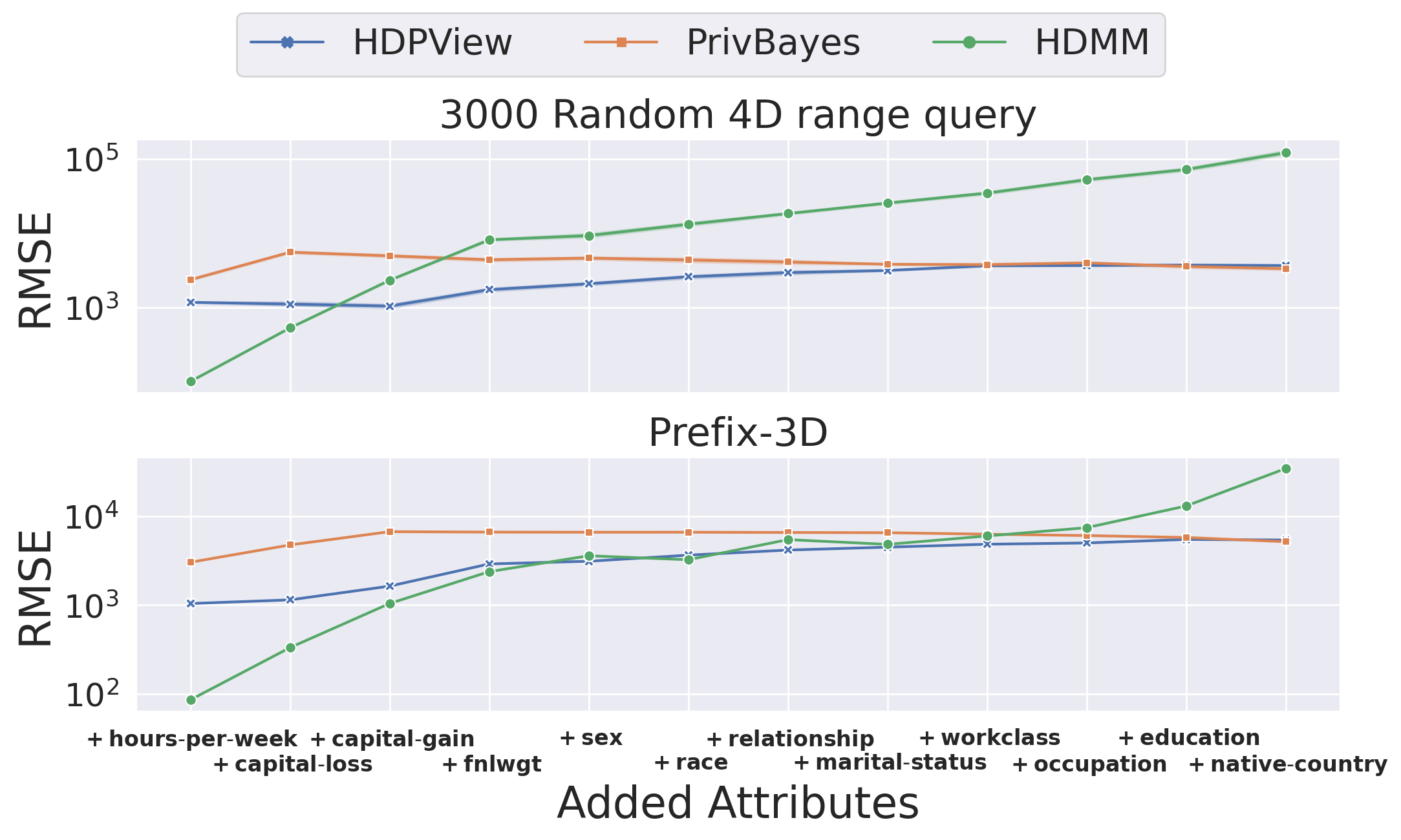}
    \caption{Changes in the performance when adding attributes to {\tt Adult} one by one in \method, PrivBayes, and HDMM.}
    \label{fig:query_add_attr}
\end{figure}

\noindent
\textbf{Comparison with PrivBayes.}
We do not consider PrivBayes a direct competitor because it is a generative model approach that does not provide any analytical reliability as described in Section \ref{sec:related_works}.
However, PrivBayes is a state-of-the-art specialized for publishing differentially private marginal queries; therefore, we compared the accuracy to demonstrate the performance of \method.
As shown in Figure \ref{fig:query_rmse}, \method is a little more accurate than PrivBayes in many cases.
However, in {\tt Adult}, PrivBayes slightly outperforms \method.
Because PrivBayes uses Bayesian network to learn the data distribution, it can fit well even to high-dimensional data as long as the distribution of the data is easily modelable.
In \method, with larger dimensionality, the PEs grow slightly because the total number of blocks increases. 
The AEs also grow since more times of random converge result in larger errors.
Thus, the total error is at least expected to increase, and the larger dimensionality may work to the advantage of PrivBayes.
Still, \method is advantageous, especially for concentrated data such as {\tt Jm1}.

We consider the reason why on {\tt Numerical-adult}, which has a smaller dimensionality than {\tt Adult}, PrivBayes is less accurate than \method is because the effective attributes for capturing the accurate marginal distributions with Bayesian network are removed.
We can see the same behavior for {\tt Small-adult}.
The following experimental results can support this.
Figure \ref{fig:query_add_attr} describes the changes in the RMSE with attributes added to {\tt Adult} one by one in two workloads, where the added attributes are shown on the horizontal axis.
Initially, \method is more accurate than PrivBayes.
As attributes are added, \method is basically robust with increasing dimensionality, but the error increases slightly.
On the other hand, interestingly, the error in PrivBayes becomes slightly smaller.

Lastly, considering \method is better in {\tt Numerical-adult} and worse in {\tt Adult}, one of the advantages of PrivBayes may be due to the increase in categorical attributes.
Because \method bisects the ordered domain space, it may be hard to effectively divide categorical attributes, which possibly worsens the accuracy in \method.

\begin{figure}[t]
    \centering
    \includegraphics[width=\hsize]{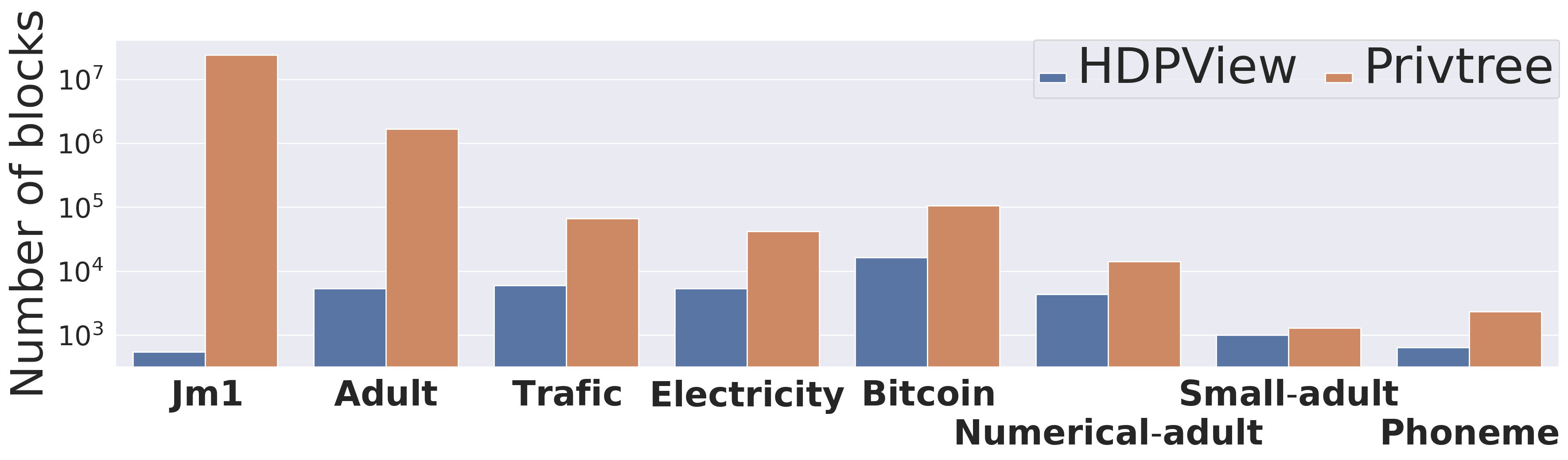}
    \caption{The number of blocks generated by \method is much lower than that generated by Privtree.} 
    \label{fig:block_size_comparison}
\end{figure}

% \begin{figure}[t]
%     \centering
%     \includegraphics[width=\hsize]{img/img_size_eps.png}
%     \caption{The data size of the p-view through \method depends on the $\#$blocks. $\#$blocks increase with higher privacy budgets.} %The data size of p-view through \method depends on $\#$blocks.}
%     \label{fig:blocknum_eps}
% \end{figure}

\subsection{Space Efficiency}
\label{sec:space_efficiency}
Our proposed p-view stores each block in a single record.
This method avoids redundancy in recording all cells that belong to the same block.
The p-view consists of blocks and values, and basically, the space complexity follows the number of blocks.
Figure \ref{fig:block_size_comparison} shows a comparison between the numbers of blocks of \method and Privtree.
While the accuracy of the counting queries of \method is higher than that of Privtree, the number of blocks generated by \method is much lower than that of Privtree, indicating that the strategy of \method avoids unnecessary splitting.
In particular, on {\tt Jm1}, \method is $4\times 10^{4}$ more efficient than Privtree.
Table \ref{tbl:space_efficiency} shows the size of the randomized views, Identity-based noisy count vector (not p-view) and p-view generated by \method at $\epsilon$=1.0.
Since \method constructs the p-view by a compact representation, it results in up to $10^{13}$ times smaller space on {\tt Adult}.

\begin{table}[t]
    \centering
    \caption{\method's p-view is space efficient (up to $10^{13}\times$).}
    \small
    \begin{tabular}{lrr}
    \toprule
    Dataset & Identity-based & \method \\
    \midrule
    {\tt Adult}       & 30.99 EB  & 3.61 MB \\
    {\tt Bitcoin}     & 1.27 TB   & 6.77 MB \\
    {\tt Electricity} & 1.11 TB   & 2.19 MB \\
    {\tt Phoneme}     & 781.34 KB & 273.59 KB \\
    \bottomrule
    \end{tabular}
    %\caption{\method is space efficient. Our proposed method reduces space consumption up to $10^{12}$ times.}
\label{tbl:space_efficiency}
\end{table}

\section{Conclusion}

We addressed the following research question: How can we construct a privacy-preserving materialized view to explore the unknown properties of the high-dimensional sensitive data?
To practically construct the p-view, we proposed a data-aware segmentation method, \method.
In our experiments, we confirmed the following desirable properties, 
% (1) effectiveness (2) space efficiency.
(1) Effectiveness: \method demonstrated smaller errors for various range counting queries in multidimensional queries. 
(2) Space efficiency: \method generates a compact representation of the p-view.
% (3) Scalability: increase in the run time of \method is sublinear with an increasing the domain size.
We believe that our method helps us explore sensitive data in the early stages of data mining while preserving data utility and privacy.

%Our method also showed better performance in classifications tasks through sampling synthesized data from the p-view.

%\clearpage

\bibliographystyle{ACM-Reference-Format}
\bibliography{ref}

\appendix
\section{Analysis for Hyperparameters}
\label{appendix:hyperparameters}

\begin{figure}[t]
    \centering
    \includegraphics[width=0.9\hsize]{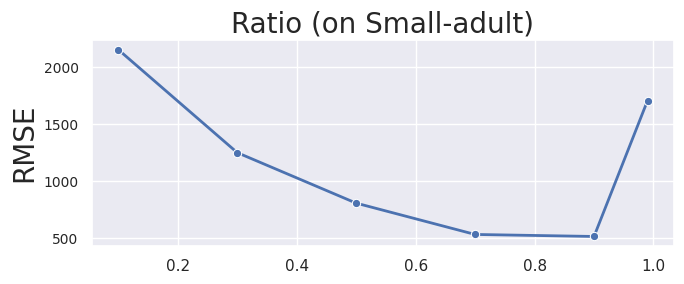}
    \caption{Effects of \method's hyperparameter $\epsilon_r/\epsilon_b$ (Ratio) on \texttt{Small-adult} dataset.}
    \label{fig:hyper_ratio}
\end{figure}

\begin{figure}[t]
    \centering
    \includegraphics[width=0.9\hsize]{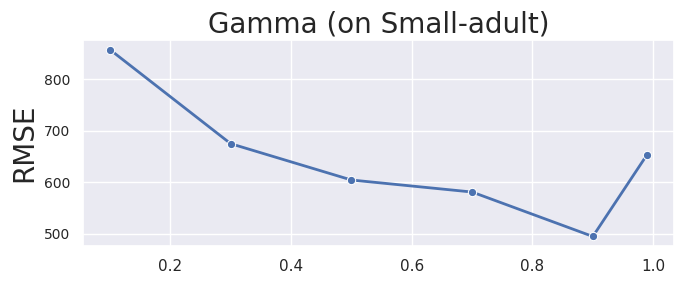}
    \caption{$\gamma$ (Gamma) on \texttt{Small-adult}.}
    \label{fig:hyper_gamma}
\end{figure}

\begin{figure}[t]
    \centering
    \includegraphics[width=0.9\hsize]{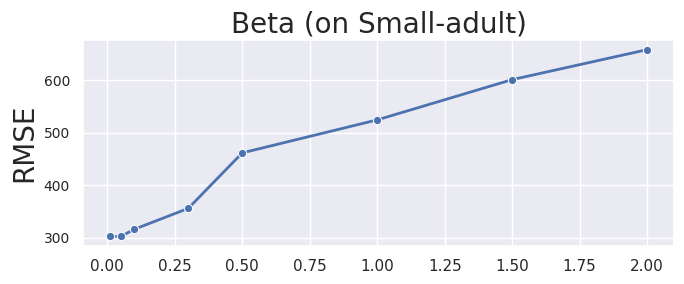}
    \caption{$\beta$ (Beta) on \texttt{Small-adult}.}
    \label{fig:hyper_beta}
\end{figure}

\begin{figure}[t]
    \centering
    \includegraphics[width=0.9\hsize]{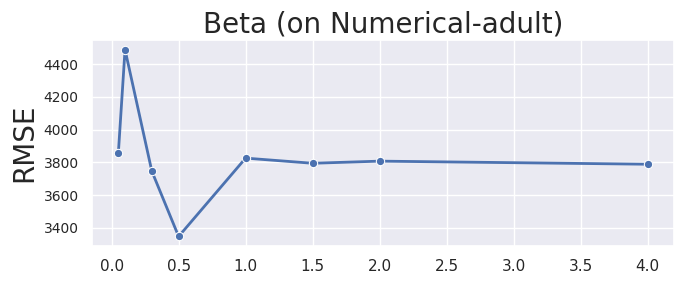}
    \caption{$\beta$ on \texttt{Numerical-adult} dataset: Optimal $\beta$ depends on dataset and our default parameter $\beta=1.2$ is somewhat conservative.}
    \label{fig:hyper_nume}
\end{figure}

\begin{figure}[t]
    \centering
    \includegraphics[width=0.9\hsize]{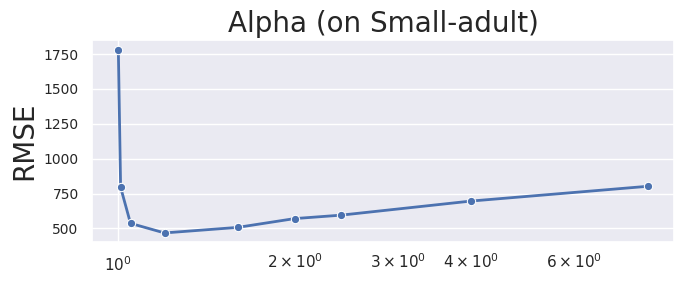}
    \caption{$\alpha$ (Alpha) on \texttt{Small-adult}.}
    \label{fig:hyper_alpha}
\end{figure}

\begin{figure}[t]
\centering
    \includegraphics[width=0.9\hsize]{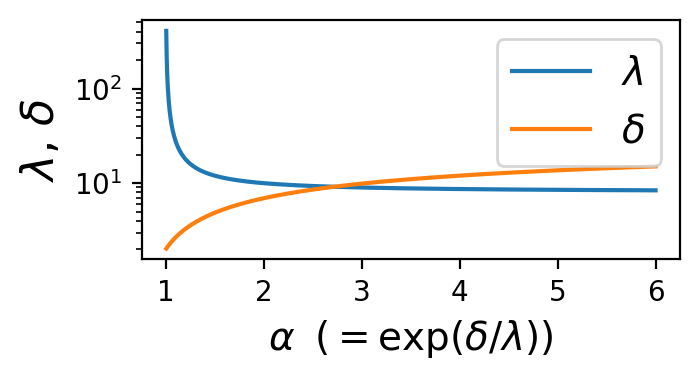}
    \caption{$\lambda$ and $\delta$ on various $\alpha$ when $\gamma \epsilon_r = 1.0$.}
    \label{fig:alpha}
\end{figure}

We provide an explanation of the hyperparameters of \method.
As shown in Algorithm \ref{algo:proposed}, \method requires four main hyperparameters, $\epsilon_r/\epsilon_b$, $\alpha$, $\beta$ and $\gamma$.
We mentioned in Section \ref{sec:exp_setup} that we fix the hyperparameters as $(\epsilon_{r}/\epsilon_{b}, \alpha, \beta, \gamma) = (0.9, 1.6, 1.2, 0.9)$ in our experiments.
Here, we provide some insights into each hyperparameter from observations of experimental results on a real-world dataset {\tt Small-adult} varying each hyperparameter.

Figures \ref{fig:hyper_ratio} \ref{fig:hyper_gamma} \ref{fig:hyper_beta} \ref{fig:hyper_alpha} show the RMSEs for \textit{3000 random 2D range query} on \texttt{Small-adult} dataset when only one of the hyperparameters varies and others are fixed as the abovementioned default.
From this result, we obtain the following insights:

\noindent
\textbf{Ratio, $\epsilon_r/\epsilon_b$}. The best accuracy is achieved when the Ratio is approximately 0.7\textasciitilde 0.9 as shown in Figure \ref{fig:hyper_ratio}. In \method, seemingly, the effect of aggregation error is larger than perturbation error, therefore we try to allocate a budget to the cutting side so that the aggregation error is smaller.

\noindent
\textbf{Gamma, $\gamma$}. As shown in Figure \ref{fig:hyper_gamma}, it was confirmed that prioritizing the\textit{random converge} to accurately determine AEs improves accuracy rather than \textit{random cut}.
However, if no budget is allocated for \textit{random cut}, the error increases, i.e., a completely random cutting strategy lose accuracy compared to appropriate our proposed \textit{random cut}.
Therefore, 0.9 is reasonable.
However, the random cut may be less less significant due to the conservative setting of the $\beta$ shown below.

\noindent
\textbf{Beta, $\beta$}. The $\beta$ is somewhat conservatively determined. 
We choose $\beta = 1.2$ because, when $\beta=1.2$, the maximum depth of \method’s bisection rarely reaches $\kappa$ on various datasets. 
As a rough guideline, if the total number of domains in a given dataset is $\Tilde{n}$, all blocks can be split at $\log_2{\Tilde{n}}$ times depth, assuming the domains are bisected exactly in the all cutting.
Thus $\kappa = 1.2* \log_2{\Tilde{n}}$ is deep enough to split all the blocks, allowing for some skewness, and all the cutting point is likely to be selected by an Exponential Mechanism rather than by random one. 
However, remember the budget for each EM is inversely proportional to $\kappa$, depending on the data set, the budget available for EM may be unnecessarily small due to the unnecessarily large setting of $\kappa$ as is the case with {\tt Small-adult} (Figure \ref{fig:hyper_beta}). 
We also show the {\tt Numerical-adult} result as another example (Figure \ref{fig:hyper_nume}). The small $\beta$ do not take full advantage of the random cut. 
Since determining the optimal $\beta$ for any dataset is impossible without additional privacy consumption, we conservatively set $\beta=1.2$ for all dataset in the experiment.

\noindent
\textbf{Alpha, $\alpha$}. $\alpha$, i.e., $\exp(\delta / \lambda)$, is valid for $\alpha > 1$.
If $\alpha$ is extremely close to 1, $\lambda$ diverges and \method does not work well because random converge causes large errors.
Because $\lambda = (\frac{3\alpha-2}{\alpha-1}) \cdot (\frac{2}{\gamma \epsilon_r})$ and $\delta = \lambda \log{\alpha}$, as $\alpha$ increases, $\lambda$ decreases and converges to 3, but $\delta$ increases.
Thus $\lambda$ and $\delta$ are trade-offs.
When $\delta$ increases, the bias of BAE increases, which also leads to a worse convergence decision.
Figure \ref{fig:alpha} plots the size of $\lambda$ and $\delta$ for various $\alpha$ when $\gamma \epsilon_r = 1.0$.
As $\alpha$ increases, the $\delta$ increases, but the decrease in $\lambda$ is small, starting from approximately 1.4.
Therefore, around $\alpha=1.4\sim1.8$ works well empirically and we use $\alpha=1.6$ as default value.

\end{document}